\documentclass[]{llncs}
\RequirePackage[colorlinks=true, citecolor=blue]{hyperref}

\usepackage{geometry}
\usepackage{tikz}
\usepackage{pgfplots}
\usetikzlibrary{decorations.markings}
\usetikzlibrary{matrix,backgrounds,folding}
\usetikzlibrary{shapes.geometric}
\pgfdeclarelayer{edgelayer}
\pgfdeclarelayer{nodelayer}
\pgfsetlayers{background,edgelayer,nodelayer,main}
\tikzstyle{every picture}=[baseline=-0.25em]
\tikzstyle{none}=[inner sep=0mm]
\tikzstyle{zxnode}=[shape=circle, minimum width=.25cm, inner sep=0.5pt, font=\scriptsize, draw=black]
\tikzstyle{gn}=[zxnode ,fill=green]
\tikzstyle{rn}=[zxnode ,fill=red]
\tikzstyle{H box}=[rectangle,fill=yellow,draw=black,xscale=1,yscale=1,font=\small,inner sep=0.75pt,minimum width=0.15cm,minimum height=0.15cm]
\tikzstyle{ug}=[regular polygon, regular polygon sides=3, fill=red,draw=black,inner sep = 0pt,minimum width=1em]
\tikzstyle{black dot}=[inner sep=0.7mm,minimum width=0pt,minimum height=0pt,fill=black,draw=black,shape=circle]
\tikzstyle{dot}=[black dot]
\tikzstyle{white dot}=[dot,fill=white]
\tikzstyle{zwcross}=[diamond, draw, fill=gray, minimum width=0em, inner sep=1.5pt]

\tikzstyle{st}=[star,star points = 5, fill=white,draw=black,inner sep = 1.2pt,line width=1.2pt]

\pgfdeclarelayer{edgelayer}
\pgfdeclarelayer{nodelayer}
\pgfsetlayers{background,edgelayer,nodelayer,main}
\tikzstyle{none}=[inner sep=0mm]
\tikzstyle{every loop}=[]

\newcommand{\half}{\begin{tikzpicture}
		\node [style=st] (0) at (0,0) {};
\end{tikzpicture}}
\newcommand{\two}{\begin{tikzpicture}
		\draw (0,0) circle (0.25) ;
\end{tikzpicture}}
\newcommand{\nmcrossI}[2]{\scalebox{0.8}{\begin{tikzpicture}
	\begin{pgfonlayer}{nodelayer}
		\node [style=none] (0) at (0.25, 0.5) {};
		\node [style=none] (1) at (0.25, -0.5) {};
		\node [style=none] (2) at (0.5, 0.5) {};
		\node [style=none] (3) at (0.5, -0.5) {};
		\node [style=none] (4) at (0, -0.5) {};
		\node [style=none] (5) at (-0.5, -0.5) {};
		\node [style=none] (6) at (-1, -0.5) {};
		\node [style=none] (7) at (-1.5, -0.5) {};
		\node [style=none] (8) at (0, 0.5) {};
		\node [style=none] (9) at (-0.5, 0.5) {};
		\node [style=none] (10) at (-1, 0.5) {};
		\node [style=none] (11) at (-1.5, 0.5) {};
		\node [style=none] (12) at (-1.25, -0.5) {$~\cdots~$};
		\node [style=none] (13) at (-1.25, 0.5) {$~\cdots~$};
		\node [style=none] (14) at (-0.2500002, 0.5) {$~\cdots~$};
		\node [style=none] (15) at (-0.2500002, -0.5) {$~\cdots~$};
		\node [style=none] (16) at (-1.25, 0.75) {$#1$};
		\node [style=none] (17) at (-0.2500002, 0.75) {$#2$};
	\end{pgfonlayer}
	\begin{pgfonlayer}{edgelayer}
		\draw (0.center) to (1.center);
		\draw (2.center) to (3.center);
		\draw [in=90, out=-90, looseness=0.75] (8.center) to (6.center);
		\draw [in=90, out=-90, looseness=0.75] (9.center) to (7.center);
		\draw [in=90, out=-90, looseness=0.75] (11.center) to (5.center);
		\draw [in=90, out=-90, looseness=0.75] (10.center) to (4.center);
	\end{pgfonlayer}
\end{tikzpicture}}}
\newcommand{\nmcross}[2]{\scalebox{0.8}{\begin{tikzpicture}
	\begin{pgfonlayer}{nodelayer}
		\node [style=none] (0) at (0.7500001, -0.5) {};
		\node [style=none] (1) at (0.2499998, -0.5) {};
		\node [style=none] (2) at (-0.2499998, -0.5) {};
		\node [style=none] (3) at (-0.7500001, -0.5) {};
		\node [style=none] (4) at (0.7500001, 0.5) {};
		\node [style=none] (5) at (0.2499998, 0.5) {};
		\node [style=none] (6) at (-0.2499998, 0.5) {};
		\node [style=none] (7) at (-0.7500001, 0.5) {};
		\node [style=none] (8) at (-0.5000002, -0.5) {$~\cdots~$};
		\node [style=none] (9) at (-0.5000002, 0.5) {$~\cdots~$};
		\node [style=none] (10) at (0.5000002, 0.5) {$~\cdots~$};
		\node [style=none] (11) at (0.5000002, -0.5) {$~\cdots~$};
		\node [style=none] (12) at (-0.5000002, 0.7499999) {$#1$};
		\node [style=none] (13) at (0.5000002, 0.7499999) {$#2$};
	\end{pgfonlayer}
	\begin{pgfonlayer}{edgelayer}
		\draw [in=90, out=-90, looseness=0.75] (4.center) to (2.center);
		\draw [in=90, out=-90, looseness=0.75] (5.center) to (3.center);
		\draw [in=90, out=-90, looseness=0.75] (7.center) to (1.center);
		\draw [in=90, out=-90, looseness=0.75] (6.center) to (0.center);
	\end{pgfonlayer}
\end{tikzpicture}}}

\usepackage{xspace}
\usepackage{mathtools}
\usepackage{calc}
\usepackage{multicol}
\usepackage{xstring}
\usepackage{xcolor}
\usepackage{amssymb}
\usepackage{amsfonts}
\usepackage{amsmath}
\usepackage{stmaryrd}

\newtheorem{thmpart}{Part}

\def\fig{}

\allowdisplaybreaks

\def \zx {\textnormal{ZX}\xspace}
\def \zxcalc {\text{ZX}-\text{Cal}\-\text{culus}\xspace}
\def \zxtcalc {\text{ZX$_{\pi/4}$}-\text{Cal}\-\text{culus}\xspace}
\def \zw {\textnormal{ZW}\xspace}
\def \zwcalc {\text{ZW}-\text{Cal}\-\text{culus}\xspace}
\def \zxt {\textnormal{ZX$_{\pi/4}$}\xspace}
\def \zwh {\textnormal{ZW$_{1/2}$}\xspace}
\def \zwhcalc {\text{ZW$_{1/2}$}-\text{Cal}\-\text{culus}\xspace}

\newcommand{\fit}[1] {\resizebox{\columnwidth}{!}{#1}}
\newcommand{\frag}[1]{$\frac{\pi}{#1}$-frag\-ment}

\newcommand{\titlerule}[1]{\begin{minipage}{\columnwidth}\begin{center}
\rule{(\textwidth-\widthof{#1})/2}{0.5pt}#1\rule{(\textwidth-\widthof{#1})/2}{0.5pt}
\end{center}
\vspace{10em}
\end{minipage}
\vspace{-10.5em}}

\newcommand{\interp}[1] {\left\llbracket #1 \right\rrbracket}
\newcommand{\interpwx}[1]{\interp{#1}_{WX}}
\newcommand{\interpxw}[1]{\interp{#1}_{XW}}
\newcommand{\dblinterp}[1]{\interp{#1}^{\natural}}

\newcommand{\annoted}[3]{{\scriptstyle #1}\left\lbrace\mathrlap{\phantom{#3}}\right.\overbrace{#3}^{#2}}
\def\piq{e^{i\frac{\pi}{4}}}

\newcommand{\callrule}[2]{\hyperlink{r:#1}{\textnormal{(#2)}}\xspace}

\newcommand{\so}{\callrule{rules}{S1}}
\newcommand{\st}{\callrule{rules}{S2}}
\newcommand{\sth}{\callrule{rules}{S3}}
\newcommand{\e}{\callrule{rules}{E}}
\newcommand{\bo}{\callrule{rules}{B1}}
\newcommand{\bt}{\callrule{rules}{B2}}
\newcommand{\kt}{\callrule{rules}{K}}
\newcommand{\eu}{\callrule{rules}{EU}}
\newcommand{\h}{\callrule{rules}{H}}
\newcommand{\supp}{\callrule{rules}{SUP}}
\newcommand{\com}{\callrule{rules}{C}}

\newcommand{\nfi}{\callrule{rules}{BW}}

\newcommand{\eq}[2][~~]{
#1
\underset{\substack{#2}}{=}
#1
}

\renewcommand*{\arraystretch}{0.8}

\newcounter{steps}%
\newcommand\step{\refstepcounter{steps}\thesteps}%
\def\thesteps{\ensuremath{\roman{steps})}}%

\title{A Complete Axiomatisation of the ZX-Calculus for Clifford+T Quantum Mechanics}        
\author{
Emmanuel Jeandel
\and Simon Perdrix
\and Renaud Vilmart
\institute{Universit\'e de Lorraine, CNRS, Inria, LORIA, F 54000 Nancy, France}
\\\email{emmanuel.jeandel@loria.fr}$\quad$
\email{simon.perdrix@loria.fr}$\quad$
\email{renaud.vilmart@loria.fr}
}
\begin{document}

\maketitle

\begin{abstract}
We introduce the first complete and approximatively universal diagrammatic language for quantum mechanics.  We make the \zxcalc, a diagrammatic language introduced by Coecke and Duncan, complete for the so-called Clifford+T quantum mechanics by adding two new axioms to the language. The completeness of the \zxcalc for Clifford+T quantum mechanics was one of the main open questions in \emph{categorical quantum mechanics}. 
We prove the completeness of the \frag{4} of the \zxcalc using the recently studied \zwcalc, a calculus dealing with integer matrices.  We also prove that  the \frag{4} of the  \zxcalc represents exactly all the matrices over some finite dimensional extension of the ring of dyadic rationals.
\end{abstract}

\section{Introduction}

The \zxcalc is a powerful graphical language for quantum reasoning and quantum computing introduced by Bob Coecke and Ross Duncan \cite{interacting}. 
The language comes with a way of interpreting any ZX-diagram as a matrix -- called the \emph{standard interpretation}. Two diagrams represent the same quantum evolution when they have the same standard interpretation. The language is also equipped with a set of axioms -- transformation rules -- which are sound, i.e. they preserve the standard interpretation. Their purpose is to explain how a diagram can be transformed into an equivalent one.

The ZX-calculus has several applications in 
quantum information processing \cite{picturing-qp} (e.g. measure\-ment-based quantum computing  \cite{mbqc,horsman2011quantum,duncan2013mbqc}, quantum codes \cite{verifying-color-code,duncan2014verifying,chancellor2016coherent,de2017zx}, foundations \cite{toy-model-graph,duncan2016hopf}), and can be used through the interactive theorem prover Quantomatic \cite{quanto,kissinger2015quantomatic}. However, the main obstacle to wider use of the ZX-calculus was the absence of a \emph{completeness} result for a \emph{universal} fragment of quantum mechanics, in order to guarantee that any true property is provable using the ZX-calculus. More precisely, 
the language would be complete if, given any two diagrams representing the same matrix, 
one could transform one diagram into the other using the axioms of the language. 
Completeness is crucial,  it means in particular that all the fundamental properties of quantum mechanics are captured by the graphical rules.

\zxcalc has been proved to be incomplete in general \cite{incompleteness}, and despite the necessary axioms that have since been identified \cite{supplementarity,gen-supp}, the language remained incomplete. 
However, several fragments of the language have been proved to be complete (\frag{2} \cite{pi_2-complete}; $\pi$-fragment \cite{pivoting}; single-qubit \frag{4} \cite{pi_4-single-qubit}), but none of them are \emph{universal} for quantum mechanics, even approximatively. In particular all quantum algorithms expressible in these fragments are efficiently simulable on a classical computer. 

As a consequence, most of the attention has been paid to find a complete axiomatisation of the \frag{4} of the \zxcalc for the Clifford+T quantum mechanics, the simplest approximatively universal fragment of quantum mechanics, which is widely used in quantum computing.

{\bf Our Approach.} In the following, we introduce the first complete axiomatisation of the ZX-calculus for Clifford+T quantum mechanics, 
thanks to the help of the \zwcalc, another graphical language -- based on the interactions of the so-called GHZ and W states \cite{ghz-w}.
 The \zwcalc has been proved to be complete \cite{zw} but its diagrams only represent matrices over $\mathbb{Z}$, and hence is not approximatively universal. We introduce the ZW$_{1/2}$-calculus, a simple extension of the \zwcalc which remains  complete and in which any matrix over the dyadic rational numbers can be represented. 
We then introduce two interpretations from the \zxcalc  to  the ZW$_{1/2}$-calculus and back. Thanks to these interpretations, we derive the completeness of the \frag 4 of the \zxcalc from the completeness of the ZW$_{1/2}$-calculus. Notice that the interpretation of ZX-diagrams (which represent complex matrices)  into ZW$_{1/2}$-diagrams (which represent dyadic rationals) requires a non trivial encoding. 
Notice also that this approach provides a completion procedure. Roughly speaking each axiom of the ZW$_{1/2}$-calculus generates an equation in the ZX-calculus: if this equation is not already provable using the existing axioms of the ZX-calculus one can treat this equality as an new axiom. A great part of the work has been to reduce all these equalities to only two additional axioms for the language.

{\bf Related works.} The first version of the present paper has been uploaded on Arxiv in May 2017. 
 In the following weeks,  Hadzihasanovic \cite{Amar} independently introduced in his PhD thesis the ZW$_{\mathbb C}$-calculus, an extension of the \zwcalc, which is universal and complete for complex matrices. Notice that the ZW$_{\mathbb C}$-calculus does not capture the peculiar properties of Clifford+T quantum mechanics, and hence the use of the ZW$_{1/2}$ remains crucial in the proof of the completeness of the ZX-calculus for this fragment. 
Based upon our work and Hadzihasanovic's, Ng and Wang \cite{NgWang} have then introduced a complete axiomatisation of the ZX-calculus for the full quantum mechanics. Their approach consists in deriving the completeness of the ZX-calculus from the completeness of the ZW$_{\mathbb C}$-calculus, using a completion procedure based on the back and forth interpretations. In \cite{JPV-universal}, we improved this result, showing that a single additional axiom is sufficient  to make the ZX-calculus complete in general, whereas 22 new axioms together with two additional generators are used in \cite{NgWang}. 
Ng and Wang uploaded afterwards a note \cite{NgWang-clifford+t} on Arxiv providing an alternative complete axiomatisation of the ZX-calculus for Clifford+T quantum mechanics. Their axiomatisation is using two additional generators and significantly more axioms  than the axiomatisation given in the present paper.

The paper is structured as follows: A \zxcalc augmented with two new axioms is presented in Section \ref{sec:ZX}. Section \ref{sec:bird} gives a general overview of the completeness proof. In Section \ref{section:zw}, we introduce an extension of the \zwcalc that deals with matrices over dyadic rational numbers $\mathbb{D} = \mathbb{Z}[1/2]$ and show its completeness. Sections \ref{section:ZXZW} and \ref{section:ZWZX} are presenting a back and forth translation between the ZX- and ZW-calculi, from which we deduce the completeness of the \zxcalc for Clifford+T quantum mechanics in section \ref{section:complete}. In Section \ref{section:exp}, we characterise the exact expressive power of the \frag{4} of the \zxcalc: the diagrams of this fragment represent exactly the matrices over $\mathbb{D}[e^{i\frac{\pi}{4}}]$. In section  \ref{sec:dicussion}, we briefly discuss the interpretation of the two new axioms of the language.

\section{\zxcalc}
\label{sec:ZX}

\subsection{Diagrams and standard interpretation}

A ZX-diagram $D:k\to l$ with $k$ inputs and $l$ outputs is generated by:
\begin{center}
\bgroup
\def\arraystretch{2.5}
{\begin{tabular}{|cc|cc|}
\hline
$R_Z^{(n,m)}(\alpha):n\to m$ & 
\InputIfFileExists{gn-alpha.tikz}{}{\input{./figures/gn-alpha.tikz}}
 & $R_X^{(n,m)}(\alpha):n\to m$ & 
\InputIfFileExists{rn-alpha.tikz}{}{\input{./figures/rn-alpha.tikz}}
\\[4ex]\hline
$H:1\to 1$ & 
\begin{tikzpicture}
	\begin{pgfonlayer}{nodelayer}
		\node [style={H box}] (0) at (0, 0) {};
		\node [style=none] (1) at (0, 0.5) {};
		\node [style=none] (2) at (0, -0.5) {};
	\end{pgfonlayer}
	\begin{pgfonlayer}{edgelayer}
		\draw (2.center) to (1.center);
	\end{pgfonlayer}
\end{tikzpicture}
}
 & $e:0\to 0$ & 
\InputIfFileExists{empty-diagram.tikz}{}{\input{./figures/empty-diagram.tikz}}
\\\hline
$\mathbb{I}:1\to 1$ & 
\begin{tikzpicture}
	\begin{pgfonlayer}{nodelayer}
		\node [style=none] (0) at (0, 0.2499999) {};
		\node [style=none] (1) at (0, -0.2499999) {};
	\end{pgfonlayer}
	\begin{pgfonlayer}{edgelayer}
		\draw (0.center) to (1.center);
	\end{pgfonlayer}
\end{tikzpicture}}
 & $\sigma:2\to 2$ & 
\InputIfFileExists{crossing.tikz}{}{\input{./figures/crossing.tikz}}
\\\hline
$\epsilon:2\to 0$ & 
\begin{tikzpicture}
	\begin{pgfonlayer}{nodelayer}
		\node [style=none] (0) at (-0.2500001, 0.2500001) {};
		\node [style=none] (1) at (0.2500001, 0.2500001) {};
	\end{pgfonlayer}
	\begin{pgfonlayer}{edgelayer}
		\draw [bend right=90, looseness=1.75] (0.center) to (1.center);
	\end{pgfonlayer}
\end{tikzpicture}}
 & $\eta:0\to 2$ & 
\begin{tikzpicture}
	\begin{pgfonlayer}{nodelayer}
		\node [style=none] (0) at (-0.2500001, -0) {};
		\node [style=none] (1) at (0.2500001, -0) {};
	\end{pgfonlayer}
	\begin{pgfonlayer}{edgelayer}
		\draw [bend left=90, looseness=1.75] (0.center) to (1.center);
	\end{pgfonlayer}
\end{tikzpicture}}
\\\hline
\end{tabular}}
\egroup\\
where $n,m\in \mathbb{N}$, $\alpha \in \mathbb{R}$, and the generator $e$ is the empty diagram.
\end{center}
and the two compositions:
\begin{itemize}
\item Spacial Composition: for any $D_1:a\to b$ and $D_2:c\to d$, $D_1\otimes D_2:a+c\to b+d$ consists in placing $D_1$ and $D_2$ side by side, $D_2$ on the right of $D_1$.
\item Sequential Composition: for any $D_1:a\to b$ and $D_2:b\to c$, $D_2\circ D_1:a\to c$ consists in placing $D_1$ on the top of $D_2$, connecting the outputs of $D_1$ to the inputs of $D_2$.
\end{itemize}

The standard interpretation of the ZX-diagrams associates to any diagram $D:n\to m$ a linear map $\interp{D}:\mathbb{C}^{2^n}\to\mathbb{C}^{2^m}$ inductively defined as follows:\\
\titlerule{$\interp{.}$}
\[ \interp{D_1\otimes D_2}:=\interp{D_1}\otimes\interp{D_2} \qquad 
\interp{D_2\circ D_1}:=\interp{D_2}\circ\interp{D_1}\qquad\interp{
\InputIfFileExists{empty-diagram.tikz}{}{\input{./figures/empty-diagram.tikz}}
~}:=\begin{pmatrix}
1
\end{pmatrix} \qquad
\interp{~
\begin{tikzpicture}
	\begin{pgfonlayer}{nodelayer}
		\node [style=none] (0) at (0, 0.2499999) {};
		\node [style=none] (1) at (0, -0.2499999) {};
	\end{pgfonlayer}
	\begin{pgfonlayer}{edgelayer}
		\draw (0.center) to (1.center);
	\end{pgfonlayer}
\end{tikzpicture}}
~~}:= \begin{pmatrix}
1 & 0 \\ 0 & 1\end{pmatrix}\]
\[\interp{~
\begin{tikzpicture}
	\begin{pgfonlayer}{nodelayer}
		\node [style={H box}] (0) at (0, 0) {};
		\node [style=none] (1) at (0, 0.5) {};
		\node [style=none] (2) at (0, -0.5) {};
	\end{pgfonlayer}
	\begin{pgfonlayer}{edgelayer}
		\draw (2.center) to (1.center);
	\end{pgfonlayer}
\end{tikzpicture}
}
~}:= \frac{1}{\sqrt{2}}\begin{pmatrix}1 & 1\\1 & -1\end{pmatrix}\qquad
\interp{
\InputIfFileExists{crossing.tikz}{}{\input{./figures/crossing.tikz}}
}:= \begin{pmatrix}
1&0&0&0\\
0&0&1&0\\
0&1&0&0\\
0&0&0&1
\end{pmatrix} \qquad
\interp{\raisebox{-0.35em}{$
\begin{tikzpicture}
	\begin{pgfonlayer}{nodelayer}
		\node [style=none] (0) at (-0.2500001, -0) {};
		\node [style=none] (1) at (0.2500001, -0) {};
	\end{pgfonlayer}
	\begin{pgfonlayer}{edgelayer}
		\draw [bend left=90, looseness=1.75] (0.center) to (1.center);
	\end{pgfonlayer}
\end{tikzpicture}}
$}}:= \begin{pmatrix}
1\\0\\0\\1
\end{pmatrix}\qquad
\interp{\raisebox{-0.25em}{$
\begin{tikzpicture}
	\begin{pgfonlayer}{nodelayer}
		\node [style=none] (0) at (-0.2500001, 0.2500001) {};
		\node [style=none] (1) at (0.2500001, 0.2500001) {};
	\end{pgfonlayer}
	\begin{pgfonlayer}{edgelayer}
		\draw [bend right=90, looseness=1.75] (0.center) to (1.center);
	\end{pgfonlayer}
\end{tikzpicture}}
$}}:= \begin{pmatrix}
1&0&0&1
\end{pmatrix}\]
\[
\interp{\begin{tikzpicture}
	\begin{pgfonlayer}{nodelayer}
		\node [style=gn] (0) at (0, -0) {$\alpha$};
	\end{pgfonlayer}
\end{tikzpicture}}:=\begin{pmatrix}1+e^{i\alpha}\end{pmatrix} \qquad
\interp{
\InputIfFileExists{gn-alpha.tikz}{}{\input{./figures/gn-alpha.tikz}}
}:=
\annoted{2^m}{2^n}{\begin{pmatrix}
  1 & 0 & \cdots & 0 & 0 \\
  0 & 0 & \cdots & 0 & 0 \\
  \vdots & \vdots & \ddots & \vdots & \vdots \\
  0 & 0 & \cdots & 0 & 0 \\
  0 & 0 & \cdots & 0 & e^{i\alpha}
 \end{pmatrix}}
~~\begin{pmatrix}n+m>0\end{pmatrix} 
\]
For any $n,m\geq 0$ and $\alpha\in\mathbb{R}$:\\
\begin{minipage}{\columnwidth}
\[\scalebox{0.9}{$\interp{
\InputIfFileExists{rn-alpha.tikz}{}{\input{./figures/rn-alpha.tikz}}
}=\interp{~
\begin{tikzpicture}
	\begin{pgfonlayer}{nodelayer}
		\node [style={H box}] (0) at (0, 0) {};
		\node [style=none] (1) at (0, 0.5) {};
		\node [style=none] (2) at (0, -0.5) {};
	\end{pgfonlayer}
	\begin{pgfonlayer}{edgelayer}
		\draw (2.center) to (1.center);
	\end{pgfonlayer}
\end{tikzpicture}
}
~}^{\otimes m}\circ \interp{
\InputIfFileExists{gn-alpha.tikz}{}{\input{./figures/gn-alpha.tikz}}
}\circ \interp{~
\begin{tikzpicture}
	\begin{pgfonlayer}{nodelayer}
		\node [style={H box}] (0) at (0, 0) {};
		\node [style=none] (1) at (0, 0.5) {};
		\node [style=none] (2) at (0, -0.5) {};
	\end{pgfonlayer}
	\begin{pgfonlayer}{edgelayer}
		\draw (2.center) to (1.center);
	\end{pgfonlayer}
\end{tikzpicture}
}
~}^{\otimes n}$}\] \\
$\left(\text{where }M^{\otimes 0}=\begin{pmatrix}1\end{pmatrix}\text{ and }M^{\otimes k}=M\otimes M^{\otimes k-1}\text{ for any }k\in \mathbb{N}^*\right)$.\\
\rule{\columnwidth}{0.5pt}
\end{minipage}\\

To simplify, the red and green nodes will be represented empty when holding a 0 angle:
\[ \scalebox{0.9}{
\InputIfFileExists{gn-empty-is-gn-zero.tikz}{}{\input{./figures/gn-empty-is-gn-zero.tikz}}
} \qquad\text{and}\qquad \scalebox{0.9}{
\InputIfFileExists{rn-empty-is-rn-zero.tikz}{}{\input{./figures/rn-empty-is-rn-zero.tikz}}
} \]
Also in order to make the diagrams a little less heavy, when $n$ copies of the same sub-diagram occur, we will use the notation $(.)^{\otimes n}$.


ZX-Diagrams are universal:
\[\forall A\in \mathbb{C}^{2^n}\times\mathbb{C}^{2^m},~~\exists D:n\to m,~~ \interp{D}=A\]
This implies dealing with an uncountable set of angles, so it is generally preferred to work with \textit{approximate} universality -- the ability to approximate any linear map with arbitrary accuracy -- in which only a finite set of angles is involved. The \frag{4} -- ZX-diagrams where all angles are multiples of $\frac{\pi}{4}$ -- is one such approximately universal fragment, whereas the \frag{2} is not \cite{clifford-not-universal}.

\subsection{Calculus}

The diagrammatic representation of a matrix is not unique in the \zxcalc. As a consequence the language comes with a set of axioms. Additionally to the axioms of the language described in Figure \ref{fig:ZX_rules}, one can:
\begin{itemize}
\item bend any wire  of a ZX-diagram at will, without changing its semantics. This paradigm -- the so-called \textbf{Only Topology Matters} -- can be derived from the following axioms:
\[{
\InputIfFileExists{bent-wire-1.tikz}{}{\input{./figures/bent-wire-1.tikz}}
}\]
\[{
\InputIfFileExists{bent-wire-2.tikz}{}{\input{./figures/bent-wire-2.tikz}}
}\]
\item apply the axioms to sub-diagrams. If $\zx\vdash D_1=D_2$ then, for any diagram $D$ with the appropriate number of inputs and outputs:
\begin{itemize}
\item $\zx\vdash D_1\circ D = D_2\circ D$
\item $\zx\vdash D\circ D_1 = D\circ D_2$
\item $\zx\vdash D_1\otimes D = D_2\otimes D$
\item $\zx\vdash D\otimes D_1 = D\otimes D_2$
\end{itemize}
where $\zx\vdash D_1 = D_2$ means that $D_1$ can be transformed into $D_2$ using the axioms of the \zxcalc. 
\end{itemize}
\begin{figure*}[!htb]
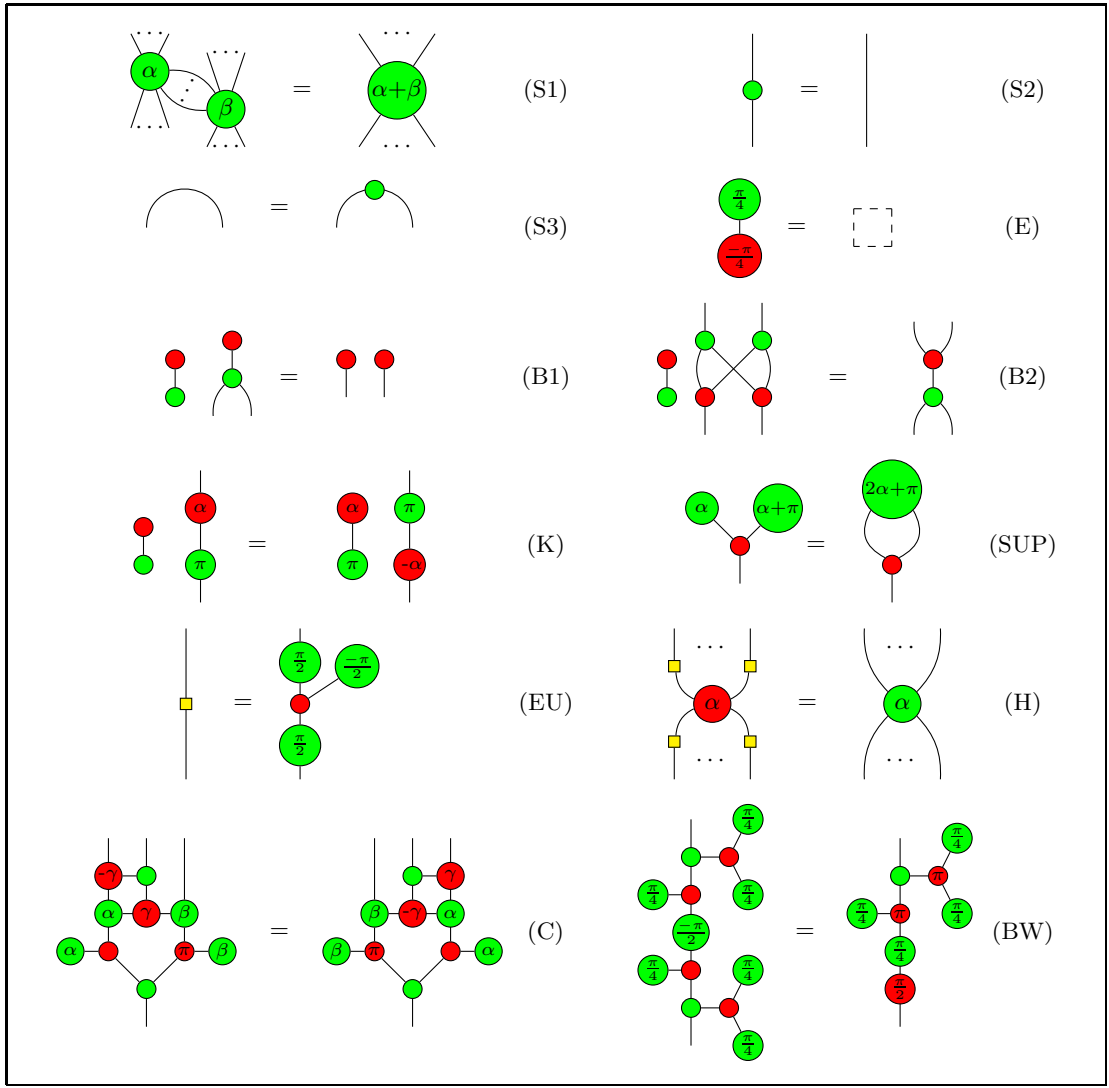

\vspace{4em}
\vfill
 \centering
 \hypertarget{r:rules}{}
  \begin{tabular}{|@{$\qquad$}ccccc@{$\qquad$}|}
   \hline
   &&$\qquad$&& \\ 
   
\InputIfFileExists{spider-1.tikz}{}{\input{./figures/spider-1.tikz}}
&(S1)&&
\InputIfFileExists{s2-simple.tikz}{}{\input{./figures/s2-simple.tikz}}
&(S2)\\
   &&&& \\
   
\InputIfFileExists{induced_compact_structure-2wire.tikz}{}{\input{./figures/induced_compact_structure-2wire.tikz}}
&(S3)&&
\InputIfFileExists{bicolor_pi_4_eq_empty.tikz}{}{\input{./figures/bicolor_pi_4_eq_empty.tikz}}
&(E)\\
   &&&&\\
   
\InputIfFileExists{b1s.tikz}{}{\input{./figures/b1s.tikz}}
&(B1) && 
\InputIfFileExists{b2s.tikz}{}{\input{./figures/b2s.tikz}}
&(B2)\\
   &&&& \\
   
\InputIfFileExists{k2s.tikz}{}{\input{./figures/k2s.tikz}}
&(K) && 
\InputIfFileExists{former-supp.tikz}{}{\input{./figures/former-supp.tikz}}
&(SUP)\\
   &&&& \\
   
\InputIfFileExists{euler-decomp-scalar-free.tikz}{}{\input{./figures/euler-decomp-scalar-free.tikz}}
&(EU) &&  
\InputIfFileExists{h2.tikz}{}{\input{./figures/h2.tikz}}
&(H)\\
   &&&& \\
   
\InputIfFileExists{commutation-of-controls-general-simplified.tikz}{}{\input{./figures/commutation-of-controls-general-simplified.tikz}}
&(C) &&
\InputIfFileExists{BW-simplified.tikz}{}{\input{./figures/BW-simplified.tikz}}
&(BW)\\
   &&&& \\
   \hline
  \end{tabular}
 \caption[]{Set of rules for the \zxcalc with scalars. All of these rules also hold when flipped upside-down, or with the colours red and green swapped. The right-hand side of (E) is an empty diagram. (...) denote zero or more wires, while (\protect\rotatebox{45}{\raisebox{-0.4em}{$\cdots$}}) denote one or more wires.\\}\vspace{4em}
\label{fig:ZX_rules}
\end{figure*}
Equalities between ZX-diagrams have the two following interesting properties:
\begin{itemize}
\item ``Colour-swapping'' (exchanging red and green dots) preserves the equality.
\item Multiplying all the angles by $-1$ preserves the equality (see Lemma \ref{lem:-1-interp}).
\end{itemize}

In the following, \zxt will denote the entire \frag4 of the ZX-Calculus with the set of rules in figure \ref{fig:ZX_rules}.

\subsection{What's new?}

We introduce in this paper a new axiomatisation of the \zxcalc. We briefly review here the differences with the previous version of the \zxcalc. 
Since we are only interested in the \frag{4} of the \zxcalc in this paper, all the axioms which are not expressible with angles multiple of $\pi/4$, like the cyclotomic supplementarity \cite{gen-supp}, are ignored. However, the rule \e, also introduced in \cite{gen-supp}, is specific to the \frag4, and hence appears in the set of rules. The two axioms (\com, \nfi) given in Figure \ref{fig:ZX_rules} are new axioms, for which we do not know any derivation using the previous axiomatisations of the language.

The two rules are of a different kind. On the one hand, \nfi is specific to the \frag4, and is hardly understandable as is. On the other hand, \com is parametrised by 3 different angles, and the rule holds whatever these angles are.

\subsection{Soundness and Completeness}

It's routine to prove the soundness of the axioms of the \zxcalc given in Figure~\ref{fig:ZX_rules}. The main result of the paper is the completeness of this axiomatisation for Clifford+T quantum mechanics:

\begin{theorem}
  \label{thm:main}
  The \frag 4 of the \zxcalc as presented in Figure~\ref{fig:ZX_rules} is
  complete: for any two diagrams $D_1, D_2$ in the \frag4 of the ZX-Calculus,
  $\interp{D_1} = \interp{D_2}$ iff $\zxt \vdash D_1 = D_2$. 
\end{theorem}

The five following sections of the paper are devoted to the proof of the theorem. A general overview of the proof is given in the next section. 

\section{A Bird's Eye View of the Proof of Theorem \ref{thm:main}}
\label{sec:bird}

The proof uses the completeness result of the \zwcalc, a
calculus dealing with matrices with integer coefficients.
The syntax and semantics of the \zwcalc are presented in
section~\ref{section:zw}.

We start by slightly changing the \zwcalc to obtain a new
language, the \zwh-calculus, that is able to express matrices with
dyadic rational coefficients, i.e. rational numbers of the form $p/2^q$.
This is done merely by adding a symbolic inverse to the scalar $2$. We
then prove that this new language is complete:

\begin{thmpart}[Proposition~\ref{prop:zwcomplete}]
\label{prop:ZWcomplete}  
The \zwh-calculus is complete: for two diagrams $D_1, D_2$ of
the \zwh-calculus, we have
  $\interp{D_1} = \interp{D_2}$ iff $\zwh \vdash D_1 = D_2$. 
\end{thmpart}
This is done in subsection~\ref{subsection:dyadic}.

\paragraph{}
We now introduce two interpretations, from \zxt to \zwh and back.

First, we provide an interpretation $\interpxw{\cdot}$
from \zxt to \zwh that transforms \zxt-dia\-grams of type $k\rightarrow l$ to
\zwh-diagrams of type $k +2 \rightarrow l+2$. This interpretation
is sound in the following sense:

\begin{thmpart}[Proposition~\ref{prop:pxwpsi}]
\label{prop:ZXZW}
  Let $D_1, D_2$ be two diagrams of the \zxt-calculus.

  Then $\interp{D_1} = \interp{D_2}$ iff $\interp{\interpxw{D_1}} = \interp{\interpxw{D_2}}$.
\end{thmpart}  
The encoding is nontrivial as the \zxtcalc expresses matrices with complex
coefficients, and the \zwh-calculus is only able to express matrices with
dyadic rational coefficients. It turns out that coefficients involved in
matrices of the \zxtcalc are actually in a vector space (more
accurately a module) of dimension $4$ over the set of dyadic rational numbers,
so that every complex coefficient will be represented by a $4 \times 4$ matrix
with dyadic rational coefficients.
This encoding is done in section \ref{section:ZXZW}.

\paragraph{}

We then provide an interpretation $\interpwx{\cdot}$ from \zwh
to \zxt. This interpretation preserves both semantics and provability:

\begin{thmpart}[Proposition~\ref{prop:interpwx} and~\ref{prop:rules-preserved}]
  \label{prop:ZWZX}

  Let $D_1, D_2$ be two diagrams of the \zwh-calculus.

  Then  $\interp{\interpwx{D_i}} = \interp{D_i}$.

  Furthermore, if $\zwh \vdash D_1 = D_2$ then $\zxt \vdash
  \interpwx{D_1} =   \interpwx{D_2}$
\end{thmpart}
This is done in section \ref{section:ZWZX}.

The composition of the two interpretations does not give back the
initial diagram (we obtain after all a diagram with two more inputs
and outputs), but we can (provably) recover it. In fact

\begin{thmpart}[Corollary~\ref{cor:composition}]
\label{prop:ZXZX}

  Let $D_1, D_2$ be a diagram of the \zxcalc.
If $\zxt \vdash \interpwx{\interpxw{D_1}} = \interpwx{\interpxw{D_2}}$ then $\zxt \vdash D_1 = D_2$.
\end{thmpart}  

Our main theorem is now obvious:
\begin{proof}[Proof of Theorem \ref{thm:main}]
Let $D_1,$ $D_2$ be two diagrams of the \zxtcalc s.t. $\interp{D_1}
  = \interp{D_2}$.
  By Part~\ref{prop:ZXZW}, $\interp{\interpxw{D_1}} = \interp{\interpxw{D_2}}$.\\
  By Part~\ref{prop:ZWcomplete}, the  \zwh-calculus is complete
  and therefore
  $\zwh \vdash \interpxw{D_1} = \interpxw{D_2}$.\\
  By Part~\ref{prop:ZWZX},
  $\zxt \vdash \interpwx{\interpxw{D_1}} = \interpwx{\interpxw{D_2}}$.\\
  By Part~\ref{prop:ZXZX} this implies $\zxt \vdash D_1 = D_2$.  
\qed\end{proof}  

This approach gives a completion procedure. It gives a set of equalities between \zxt-diagrams whose derivability proves the completeness of the language. Hence, the new rules of the \zxtcalc we introduced have obviously been chosen for Parts~\ref{prop:ZXZX} and \ref{prop:ZWZX} to hold.
However they have been greatly simplified from what one can obtain using the approach naively.

\section{\zwcalc}
\label{section:zw}
\subsection{Diagrams and Standard Interpretation}

The \zwcalc has been introduced by Amar Hadzi\-hasa\-novic in 2015 \cite{zw} and is based on the GHZ/W-Calculus \cite{ghz-w}. We will present here the expanded version of this calculus. To stay coherent with the previous definition of the \zxcalc, we will assume that the time flows from top to bottom -- which is the opposite of the original definition in the ZW-Calculus. It has the following finite set of generators:
\[T_e=\left\lbrace 

\begin{tikzpicture}
	\begin{pgfonlayer}{nodelayer}
		\node [style={white dot}] (0) at (0, -0) {};
		\node [style=none] (1) at (0, 0.5) {};
		\node [style=none] (2) at (0, -0.5) {};
	\end{pgfonlayer}
	\begin{pgfonlayer}{edgelayer}
		\draw (1.center) to (2.center);
	\end{pgfonlayer}
\end{tikzpicture}}
~,~
\InputIfFileExists{Z-2-1.tikz}{}{\input{./figures/Z-2-1.tikz}}
~,~
\begin{tikzpicture}
	\begin{pgfonlayer}{nodelayer}
		\node [style=dot] (0) at (0, -0) {};
		\node [style=none] (1) at (0, 0.5) {};
		\node [style=none] (2) at (0, -0.5) {};
	\end{pgfonlayer}
	\begin{pgfonlayer}{edgelayer}
		\draw (1.center) to (2.center);
	\end{pgfonlayer}
\end{tikzpicture}}
~,~
\begin{tikzpicture}
	\begin{pgfonlayer}{nodelayer}
		\node [style=dot] (0) at (0, 0) {};
		\node [style=none] (1) at (0, 0.5000001) {};
		\node [style=none] (2) at (-0.25, -0.5) {};
		\node [style=none] (3) at (0.25, -0.5) {};
	\end{pgfonlayer}
	\begin{pgfonlayer}{edgelayer}
		\draw (0) to (2.center);
		\draw (3.center) to (0);
		\draw (0) to (1.center);
	\end{pgfonlayer}
\end{tikzpicture}}
~,\scalebox{1}{
\begin{tikzpicture}
	\begin{pgfonlayer}{nodelayer}
		\node [style=none] (0) at (0, 0.2499999) {};
		\node [style=none] (1) at (0, -0.2499999) {};
	\end{pgfonlayer}
	\begin{pgfonlayer}{edgelayer}
		\draw (0.center) to (1.center);
	\end{pgfonlayer}
\end{tikzpicture}}
}~,\raisebox{-0.3em}{
\begin{tikzpicture}
	\begin{pgfonlayer}{nodelayer}
		\node [style=none] (0) at (-0.2500001, 0.2500001) {};
		\node [style=none] (1) at (0.2500001, 0.2500001) {};
	\end{pgfonlayer}
	\begin{pgfonlayer}{edgelayer}
		\draw [bend right=90, looseness=1.75] (0.center) to (1.center);
	\end{pgfonlayer}
\end{tikzpicture}}
}~,\raisebox{-0.4em}{
\begin{tikzpicture}
	\begin{pgfonlayer}{nodelayer}
		\node [style=none] (0) at (-0.2500001, -0) {};
		\node [style=none] (1) at (0.2500001, -0) {};
	\end{pgfonlayer}
	\begin{pgfonlayer}{edgelayer}
		\draw [bend left=90, looseness=1.75] (0.center) to (1.center);
	\end{pgfonlayer}
\end{tikzpicture}}
}~,~
\InputIfFileExists{crossing.tikz}{}{\input{./figures/crossing.tikz}}
~,~
\InputIfFileExists{zw-cross.tikz}{}{\input{./figures/zw-cross.tikz}}
~,~
\InputIfFileExists{empty-diagram.tikz}{}{\input{./figures/empty-diagram.tikz}}
~
\right\rbrace\]
and diagrams are created thanks to the same two -- spacial and sequential -- compositions.

As for the \zxcalc, we define a standard interpretation, that associates to any diagram of the \zwcalc $D$ with $n$ inputs and $m$ outputs, a linear map $\interp{D}:\mathbb{Z}^{2^n}\to\mathbb{Z}^{2^m}$, inductively defined as:\\
\titlerule{$\interp{.}$}
$$ \interp{D_1\otimes D_2}:=\interp{D_1}\otimes\interp{D_2} \qquad 
\interp{D_2\circ D_1}:=\interp{D_2}\circ\interp{D_1}\qquad
\interp{
\InputIfFileExists{empty-diagram.tikz}{}{\input{./figures/empty-diagram.tikz}}
~}:=\begin{pmatrix}1\end{pmatrix} \quad
\interp{~
\begin{tikzpicture}
	\begin{pgfonlayer}{nodelayer}
		\node [style=none] (0) at (0, 0.2499999) {};
		\node [style=none] (1) at (0, -0.2499999) {};
	\end{pgfonlayer}
	\begin{pgfonlayer}{edgelayer}
		\draw (0.center) to (1.center);
	\end{pgfonlayer}
\end{tikzpicture}}
~~}:= \begin{pmatrix}
1 & 0 \\ 0 & 1\end{pmatrix}$$
$$ 
\interp{
\InputIfFileExists{crossing.tikz}{}{\input{./figures/crossing.tikz}}
}:= \begin{pmatrix}
1&0&0&0\\
0&0&1&0\\
0&1&0&0\\
0&0&0&1
\end{pmatrix} \qquad
\interp{
\InputIfFileExists{zw-cross.tikz}{}{\input{./figures/zw-cross.tikz}}
}:= \begin{pmatrix}
1&0&0&0\\
0&0&1&0\\
0&1&0&0\\
0&0&0&-1
\end{pmatrix}\qquad
\interp{\raisebox{-0.4em}{$
\begin{tikzpicture}
	\begin{pgfonlayer}{nodelayer}
		\node [style=none] (0) at (-0.2500001, -0) {};
		\node [style=none] (1) at (0.2500001, -0) {};
	\end{pgfonlayer}
	\begin{pgfonlayer}{edgelayer}
		\draw [bend left=90, looseness=1.75] (0.center) to (1.center);
	\end{pgfonlayer}
\end{tikzpicture}}
$}}:= \begin{pmatrix}
1\\0\\0\\1
\end{pmatrix}\qquad
\interp{\raisebox{-0.3em}{$
\begin{tikzpicture}
	\begin{pgfonlayer}{nodelayer}
		\node [style=none] (0) at (-0.2500001, 0.2500001) {};
		\node [style=none] (1) at (0.2500001, 0.2500001) {};
	\end{pgfonlayer}
	\begin{pgfonlayer}{edgelayer}
		\draw [bend right=90, looseness=1.75] (0.center) to (1.center);
	\end{pgfonlayer}
\end{tikzpicture}}
$}}:= \begin{pmatrix}
1&0&0&1
\end{pmatrix}$$
$$
\interp{
\begin{tikzpicture}
	\begin{pgfonlayer}{nodelayer}
		\node [style=dot] (0) at (0, -0) {};
		\node [style=none] (1) at (0, 0.5) {};
		\node [style=none] (2) at (0, -0.5) {};
	\end{pgfonlayer}
	\begin{pgfonlayer}{edgelayer}
		\draw (1.center) to (2.center);
	\end{pgfonlayer}
\end{tikzpicture}}
}:= \begin{pmatrix}
0&1\\1&0
\end{pmatrix} \qquad 
\interp{
\begin{tikzpicture}
	\begin{pgfonlayer}{nodelayer}
		\node [style=dot] (0) at (0, 0) {};
		\node [style=none] (1) at (0, 0.5000001) {};
		\node [style=none] (2) at (-0.25, -0.5) {};
		\node [style=none] (3) at (0.25, -0.5) {};
	\end{pgfonlayer}
	\begin{pgfonlayer}{edgelayer}
		\draw (0) to (2.center);
		\draw (3.center) to (0);
		\draw (0) to (1.center);
	\end{pgfonlayer}
\end{tikzpicture}}
}:= \begin{pmatrix}
0&1\\1&0\\1&0\\0&0
\end{pmatrix}\qquad
\interp{
\begin{tikzpicture}
	\begin{pgfonlayer}{nodelayer}
		\node [style={white dot}] (0) at (0, -0) {};
		\node [style=none] (1) at (0, 0.5) {};
		\node [style=none] (2) at (0, -0.5) {};
	\end{pgfonlayer}
	\begin{pgfonlayer}{edgelayer}
		\draw (1.center) to (2.center);
	\end{pgfonlayer}
\end{tikzpicture}}
}:= \begin{pmatrix}
1&0\\0&-1
\end{pmatrix} \qquad 
\interp{
\InputIfFileExists{Z-2-1.tikz}{}{\input{./figures/Z-2-1.tikz}}
}:= \begin{pmatrix}
1&0&0&0\\0&0&0&-1
\end{pmatrix}
$$
\rule{\columnwidth}{0.5pt}
This map is obviously different from the one of the \zxcalc -- the domain is different -- but we will use the same notation.

\begin{remark}
The symbols used for the generators have be altered from the original \zwcalc in order to make it more compatible with the \zxcalc.
\end{remark}

\begin{lemma}[\cite{zw}]
ZW-Diagrams are universal for matrices of $\mathbb{Z}^{2^n}\times\mathbb{Z}^{2^m}$:
\[\forall A\in \mathbb{Z}^{2^n}\times\mathbb{Z}^{2^m},~~\exists D:n\to m,~~ \interp{D}=A\]
\end{lemma}

\subsection{Calculus}

\begin{figure*}[!bt]
\def\scale{0.75}
\centering
\begin{tabular}{|ccc|}
\hline
&&\\
\scalebox{\scale}{
\InputIfFileExists{ZW-rule-0-no-braid.tikz}{}{\input{./figures/ZW-rule-0-no-braid.tikz}}
} & $\qquad$ & \scalebox{\scale}{
\InputIfFileExists{ZW-rule-1.tikz}{}{\input{./figures/ZW-rule-1.tikz}}
} \\
&&\\
\scalebox{\scale}{
\InputIfFileExists{ZW-rule-2.tikz}{}{\input{./figures/ZW-rule-2.tikz}}
} && \scalebox{\scale}{
\InputIfFileExists{ZW-rule-3.tikz}{}{\input{./figures/ZW-rule-3.tikz}}
} \\
&&\\
 \scalebox{\scale}{
\InputIfFileExists{ZW-rule-4.tikz}{}{\input{./figures/ZW-rule-4.tikz}}
} && \scalebox{\scale}{
\InputIfFileExists{ZW-rule-5-no-braid.tikz}{}{\input{./figures/ZW-rule-5-no-braid.tikz}}
} \\
 &&\\
 \multicolumn{3}{|c|}{\scalebox{\scale}{
\InputIfFileExists{ZW-rule-6-no-braid.tikz}{}{\input{./figures/ZW-rule-6-no-braid.tikz}}
} $\qquad\qquad$ \scalebox{\scale}{
\InputIfFileExists{ZW-rule-X-no-braid.tikz}{}{\input{./figures/ZW-rule-X-no-braid.tikz}}
}}\\
 &&\\
 \multicolumn{3}{|c|}{
 \scalebox{\scale}{
\InputIfFileExists{ZW-rule-7-no-braid.tikz}{}{\input{./figures/ZW-rule-7-no-braid.tikz}}
} $\qquad\qquad$ \scalebox{\scale}{
\InputIfFileExists{reidmeister-3.tikz}{}{\input{./figures/reidmeister-3.tikz}}
}} \\
 &&\\
\hline
\end{tabular}
\caption{Set of rules for the \zwcalc.}
 \label{fig:ZW_rules}
\end{figure*}

The \zwcalc comes with a \emph{complete} set of rules ZW that is given in figure \ref{fig:ZW_rules}.
Here again, the paradigm \emph{Only Topology Matters} applies except for 
\InputIfFileExists{zw-cross.tikz}{}{\input{./figures/zw-cross.tikz}}
 where the order of inputs and outputs is important. It gives sense to nodes that are not directly given in $T_e$, e.g.:
\[
\InputIfFileExists{zw-bent-wire-example.tikz}{}{\input{./figures/zw-bent-wire-example.tikz}}
\]
All these rules are sound. We use the same notation $\vdash$ as defined in section \ref{sec:ZX}, and we can still apply the rewrite rules to subdiagrams. 
In the following we may use the shortcuts:
\begin{center}

\InputIfFileExists{black-dot-0-1.tikz}{}{\input{./figures/black-dot-0-1.tikz}}
 and\hspace{1em} 
\InputIfFileExists{white-dot-0-1.tikz}{}{\input{./figures/white-dot-0-1.tikz}}

\end{center}

\subsection{Extension to Dyadic Matrices}
\label{subsection:dyadic}

We define an extension of the \zwcalc by adding a new node that represents $\frac{1}{2}$ and binding it to the calculus with an additional rule.

\begin{definition}
We define the \zwhcalc as the extension of the \zwcalc such as:
\[\left\lbrace
\begin{array}{l}
 T_{1/2}=T_e\cup\{\half\}\\
\zwh=\zw\cup\left\lbrace
\InputIfFileExists{additional-ZW-rule.tikz}{}{\input{./figures/additional-ZW-rule.tikz}}
~\right\rbrace
\end{array}
\right.\]
The standard interpretation of a diagram $D:n\to m$ is now a matrix
$\interp{D}:\mathbb{D}^{2^n}\to\mathbb{D}^{2^m}$ over the
ring $\mathbb{D} = \mathbb{Z}[1/2]$ of dyadic rationals and is given by
the standard interpretation of the \zwcalc extended with $\interp{\half}:=\begin{pmatrix}\frac{1}{2}\end{pmatrix}$.
\end{definition}

\begin{proposition}
  \label{prop:zwcomplete}
  The \zwh is sound and complete:
  For two diagrams $D_1, D_2$ of the \zwh-calculus,
  $\interp{D_1} = \interp{D_2}$ iff $\zwh \vdash D_1 = D_2$.
\end{proposition}
\begin{proof}
  Soundness is obvious.
  
Now let $D_1$ and $D_2$ be two diagrams of the \zwh-Calculus such that $\interp{D_1}=\interp{D_2}$.
We can rewrite $D_1$ and $D_2$ as $D_i = d_i\otimes(\half)^{\otimes   n_i}$ for some integers $n_i$ and diagrams $d_i$ of the \zwcalc that do not use the $\half$ symbol.

From the new introduced rule, we get that  $\zwh\vdash d_i=D_i\otimes \left(\two\right)^{\otimes n_i}$. W.l.o.g. assume $n_1\leq n_2$. Then $\interp{d_1\otimes \left(\two\right)^{\otimes n_2-n_1}}=2^{n_2-n_1}\interp{d_1}=2^{n_2}\interp{D_1} = \interp{d_2}$. Since $d_1$ and $d_2$ are ZW-diagrams and have the same interpretation, thanks to the completeness of the \zwcalc, $\zwh\vdash d_1\otimes \left(\two\right)^{\otimes n_2-n_1}=d_2$, which means $\zwh\vdash D_1=D_2$ by applying $n_2$ times the new rule on both sides of the equality.
\qed\end{proof}


\section{From \zxt to \zwh-Diagrams}
\label{section:ZXZW}

In this section we explain how to encode diagrams of the \zxtcalc
into diagrams of the \zwh-Calculus.
The main difficulty is of course that the former represents matrices
with complex coefficients and the latter matrices with dyadic rational
coefficients.
We use for this classical results of algebra that we summarize in the
next subsection.

\subsection{From $\boldsymbol{\mathbb{Q}[\piq]}$ to $\boldsymbol{\mathbb{Q}}$}
All results used in the next two sections are standard in field
theory, see e.g. \cite{roman}.
Let $R \subseteq \mathbb{C}$ be a (commutative) ring and $\alpha \in \mathbb{C}$.
By $R[\alpha]$ we denote the smallest subring of $\mathbb{C}$ that contains both $R$ and $\alpha$.

Of primary importance will be the ring $\mathbb{Q}[\piq]$, as all terms of the $\pi/4$ fragment of the \zxcalc have interpretations as matrices in this ring.
This is clear for all terms except possibly for $\sqrt{2}$, but $\sqrt{2} = \piq - (\piq)^3$.

If $\alpha$ is algebraic, it is well known that $\mathbb{Q}[\alpha]$ is a field. When $F \subseteq F'$ are two fields, $F'$ can be seen as a vector space (actually an algebra) over $F$. Its dimension is denoted $[F':F]$ and we say that $F'$ is an extension of $F$ of degree $[F':F]$.
In the specific case of $\mathbb{Q}[\alpha]$, its dimension over $\mathbb{Q}$ is exactly the degree of the minimal polynomial over $\mathbb{Q}$ of $\alpha$.
Notice that the minimal polynomial of a $n-$th primitive root of the
unity is $\phi(n)$ where $\phi$ is Euler's totient function. 

In our case, $\piq$ is a eighth primitive root of the unity, so that $\mathbb{Q}[\piq]$ is a vector space of dimension $4$, one basis being given by $1 ,\piq,(\piq)^2, (\piq)^3$.
In particular:
\begin{proposition}
  Every element of $\mathbb{Q}[\piq]$ can be written in a unique way 
 $a + b \piq + c (\piq)^2 + d (\piq)^3$ for some rationals numbers $a,b,c,d$.
\end{proposition}

For $x \in \mathbb{Q}[\piq]$, let $\psi(x)$ be the function defined by  $\psi(x) = y \mapsto x y$. For each $x$, $\psi(x)$ is a linear map and therefore can be given by a $4\times 4$ matrix in the basis $(\piq)^3,(\piq)^2, \piq, 1$.
$\psi(1)$ is of course the identity matrix and
\[
\psi(\piq) = M = \begin{pmatrix}
  0 & 1 & 0 & 0 \\
  0 & 0 & 1 & 0 \\
  0 & 0 & 0 & 1 \\
  -1 & 0 & 0 & 0 \\
  \end{pmatrix}
\]
Notice that $M^t$ is the companion matrix of the polynomial $X^4+1$ which characterises $\piq$ as an algebraic number.

\begin{proposition}
  \label{prop:homomorphism}
  The map:
  \[  \psi: a+b\piq +c(\piq)^2 + d(\piq)^3 \mapsto a I_4 + b M + c M^2 +d M ^3\]
  is a  homomorphism of $\mathbb{Q}$-algebras from  $\mathbb{Q}[\piq]$ to $M_4(\mathbb{Q})$
\end{proposition}
This homomorphism has a left-inverse. Indeed, let
\[
\theta = \begin{pmatrix}
  1 \\
  \piq \\
  (\piq)^2 \\
  (\piq)^3 \\
\end{pmatrix}
\]
Then $\psi(x)\theta = x\theta$.

With this morphism, we can see elements of $\mathbb{Q}[\piq]$ as matrices over $\mathbb{Q}$.

Of course we can do the same with matrices over $\mathbb{Q}[\piq]$.

\begin{definition}
  Define:
  \[  \psi: A+B\piq +C(\piq)^2 + D(\piq)^3 \mapsto A\otimes  I_4 + B\otimes M + C\otimes M^2 +D\otimes M ^3\]
  $\psi$ is injective and maps a matrix over $\mathbb{Q}[\piq]$ of dimension $n\times m$ to a matrix over $\mathbb{Q}$ of dimension $4n\times 4m$.
\end{definition}  
We use the same notation  $\psi$ as before, as the definitions are equivalent for one-by-one matrices (i.e. scalars).

It is easy to see that Proposition~\ref{prop:homomorphism} holds for the extended $\psi$ in the sense that $\psi(qA) = q\psi(A)$ for $q$ rational, $\psi(A+B) = \psi(A) + \psi(B)$, $\psi(AB) = \psi(A)\psi(B)$ whenever this makes sense.

Notice however that $\psi(A\otimes B)$ is not $\psi(A) \otimes \psi(B)$.

As before, $\psi$ has a left inverse, as evidenced by:
\begin{proposition}
\label{prop:left-inverse}
  For all matrices $X$ of dimension $n \times m$, $ \psi(X)(I_m \otimes \theta) = X \otimes \theta$
\end{proposition}  

While it is true that all coefficients of the standard interpretation of the $\pi/4$ fragment are in $\mathbb{Q}[\piq]$, we can be more precise.

Let $\mathbb{D} = \mathbb{Z}[1/2]$ be the set of all dyadic rational numbers, i.e. rational numbers of the form $p/2^n$.

It is easy to see that any element of $\mathbb{D}[\piq]$ can be written in a unique way 
$a + b \piq + c (\piq)^2 + d (\piq)^3$ for some dyadic rational numbers $a,b,c,d$. (It is NOT a consequence of the similar statement for $\mathbb{Q}$. We have to use here the additional property that $\piq$ is not only an algebraic number, but also an algebraic \emph{integer}).

Then it is clear that actually all coefficients of the $\pi/4$ fragment of the \zxcalc are in $\mathbb{D}[\piq]$.
As $\mathbb{D} \subset \mathbb{Q}$ all we said before still holds, and we actually obtain with $\psi$ a map from matrices over $\mathbb{D}[\piq]$ to matrices over $\mathbb{D}$.

\subsection{Interpretation}


Based on the previous discussion, we define an interpretation
$\interpxw{.}$ from \zxt-diagrams to \zwh-diagrams as follows:\\
\titlerule{$\interpxw{.}$}\vspace{-1.5em}
\begin{multicols}{3}
\[ 
\InputIfFileExists{empty-diagram.tikz}{}{\input{./figures/empty-diagram.tikz}}
~~\mapsto~~ 
\begin{tikzpicture}
	\begin{pgfonlayer}{nodelayer}
		\node [style=none] (0) at (0, 0.2499999) {};
		\node [style=none] (1) at (0, -0.2499999) {};
	\end{pgfonlayer}
	\begin{pgfonlayer}{edgelayer}
		\draw (0.center) to (1.center);
	\end{pgfonlayer}
\end{tikzpicture}}
\quad
\begin{tikzpicture}
	\begin{pgfonlayer}{nodelayer}
		\node [style=none] (0) at (0, 0.2499999) {};
		\node [style=none] (1) at (0, -0.2499999) {};
	\end{pgfonlayer}
	\begin{pgfonlayer}{edgelayer}
		\draw (0.center) to (1.center);
	\end{pgfonlayer}
\end{tikzpicture}}
 \]\vfill
\[ 
\begin{tikzpicture}
	\begin{pgfonlayer}{nodelayer}
		\node [style=none] (0) at (0, 0.2499999) {};
		\node [style=none] (1) at (0, -0.2499999) {};
	\end{pgfonlayer}
	\begin{pgfonlayer}{edgelayer}
		\draw (0.center) to (1.center);
	\end{pgfonlayer}
\end{tikzpicture}}
~~\mapsto~~ 
\begin{tikzpicture}
	\begin{pgfonlayer}{nodelayer}
		\node [style=none] (0) at (0, 0.2499999) {};
		\node [style=none] (1) at (0, -0.2499999) {};
	\end{pgfonlayer}
	\begin{pgfonlayer}{edgelayer}
		\draw (0.center) to (1.center);
	\end{pgfonlayer}
\end{tikzpicture}}
\quad
\begin{tikzpicture}
	\begin{pgfonlayer}{nodelayer}
		\node [style=none] (0) at (0, 0.2499999) {};
		\node [style=none] (1) at (0, -0.2499999) {};
	\end{pgfonlayer}
	\begin{pgfonlayer}{edgelayer}
		\draw (0.center) to (1.center);
	\end{pgfonlayer}
\end{tikzpicture}}
\quad
\begin{tikzpicture}
	\begin{pgfonlayer}{nodelayer}
		\node [style=none] (0) at (0, 0.2499999) {};
		\node [style=none] (1) at (0, -0.2499999) {};
	\end{pgfonlayer}
	\begin{pgfonlayer}{edgelayer}
		\draw (0.center) to (1.center);
	\end{pgfonlayer}
\end{tikzpicture}}
 \]\vfill
\[ 
\begin{tikzpicture}
	\begin{pgfonlayer}{nodelayer}
		\node [style=none] (0) at (-0.2500001, 0.2500001) {};
		\node [style=none] (1) at (0.2500001, 0.2500001) {};
	\end{pgfonlayer}
	\begin{pgfonlayer}{edgelayer}
		\draw [bend right=90, looseness=1.75] (0.center) to (1.center);
	\end{pgfonlayer}
\end{tikzpicture}}
~~\mapsto~~ 
\begin{tikzpicture}
	\begin{pgfonlayer}{nodelayer}
		\node [style=none] (0) at (-0.2500001, 0.2500001) {};
		\node [style=none] (1) at (0.2500001, 0.2500001) {};
	\end{pgfonlayer}
	\begin{pgfonlayer}{edgelayer}
		\draw [bend right=90, looseness=1.75] (0.center) to (1.center);
	\end{pgfonlayer}
\end{tikzpicture}}
~~
\begin{tikzpicture}
	\begin{pgfonlayer}{nodelayer}
		\node [style=none] (0) at (0, 0.2499999) {};
		\node [style=none] (1) at (0, -0.2499999) {};
	\end{pgfonlayer}
	\begin{pgfonlayer}{edgelayer}
		\draw (0.center) to (1.center);
	\end{pgfonlayer}
\end{tikzpicture}}
\quad
\begin{tikzpicture}
	\begin{pgfonlayer}{nodelayer}
		\node [style=none] (0) at (0, 0.2499999) {};
		\node [style=none] (1) at (0, -0.2499999) {};
	\end{pgfonlayer}
	\begin{pgfonlayer}{edgelayer}
		\draw (0.center) to (1.center);
	\end{pgfonlayer}
\end{tikzpicture}}
\]\vfill
\[ 
\begin{tikzpicture}
	\begin{pgfonlayer}{nodelayer}
		\node [style=none] (0) at (-0.2500001, -0) {};
		\node [style=none] (1) at (0.2500001, -0) {};
	\end{pgfonlayer}
	\begin{pgfonlayer}{edgelayer}
		\draw [bend left=90, looseness=1.75] (0.center) to (1.center);
	\end{pgfonlayer}
\end{tikzpicture}}
~~\mapsto~~ 
\begin{tikzpicture}
	\begin{pgfonlayer}{nodelayer}
		\node [style=none] (0) at (-0.2500001, -0) {};
		\node [style=none] (1) at (0.2500001, -0) {};
	\end{pgfonlayer}
	\begin{pgfonlayer}{edgelayer}
		\draw [bend left=90, looseness=1.75] (0.center) to (1.center);
	\end{pgfonlayer}
\end{tikzpicture}}
~~\raisebox{0.5em}{
\begin{tikzpicture}
	\begin{pgfonlayer}{nodelayer}
		\node [style=none] (0) at (0, 0.2499999) {};
		\node [style=none] (1) at (0, -0.2499999) {};
	\end{pgfonlayer}
	\begin{pgfonlayer}{edgelayer}
		\draw (0.center) to (1.center);
	\end{pgfonlayer}
\end{tikzpicture}}

\begin{tikzpicture}
	\begin{pgfonlayer}{nodelayer}
		\node [style=none] (0) at (0, 0.2499999) {};
		\node [style=none] (1) at (0, -0.2499999) {};
	\end{pgfonlayer}
	\begin{pgfonlayer}{edgelayer}
		\draw (0.center) to (1.center);
	\end{pgfonlayer}
\end{tikzpicture}}
} \]\vfill
\[
\begin{tikzpicture}
	\begin{pgfonlayer}{nodelayer}
		\node [style={H box}] (0) at (0, 0) {};
		\node [style=none] (1) at (0, 0.5) {};
		\node [style=none] (2) at (0, -0.5) {};
	\end{pgfonlayer}
	\begin{pgfonlayer}{edgelayer}
		\draw (2.center) to (1.center);
	\end{pgfonlayer}
\end{tikzpicture}
}
~~\mapsto~~
\InputIfFileExists{hadamard-interpretation-2-no-braid.tikz}{}{\input{./figures/hadamard-interpretation-2-no-braid.tikz}}
\]\vfill
\[
\begin{tikzpicture}
	\begin{pgfonlayer}{nodelayer}
		\node [style=gn] (0) at (0, -0) {$\frac{\pi}{4}$};
		\node [style=none] (1) at (0, 0.5000001) {};
		\node [style=none] (2) at (0, -0.5000001) {};
	\end{pgfonlayer}
	\begin{pgfonlayer}{edgelayer}
		\draw (1.center) to (2.center);
	\end{pgfonlayer}
\end{tikzpicture}}
~~\mapsto~~
\InputIfFileExists{gn-pi_4-interpretation-2-no-braid.tikz}{}{\input{./figures/gn-pi_4-interpretation-2-no-braid.tikz}}
\]
\end{multicols}

\begin{align*}
\forall D_1&:n\to n',~\forall D_2:m\to m':\\
& D_1\circ D_2 \mapsto \interpxw{D_1}\circ \interpxw{D_2}\qquad(\text{if } m'=n)\\
& D_1\otimes D_2 \mapsto\left( \mathbb{I}^{\otimes n'}\otimes\interpxw{D_2}\right)\circ\left(\nmcrossI{m}{n'}\right)\circ \left(\mathbb{I}^{\otimes m}\otimes\interpxw{D_1}\right)\circ\left(\nmcrossI{n}{m}\right)
\end{align*}

\[
\InputIfFileExists{gn-kpi_4.tikz}{}{\input{./figures/gn-kpi_4.tikz}}
\mapsto \left(
\InputIfFileExists{gn-0-1-m-interpretation.tikz}{}{\input{./figures/gn-0-1-m-interpretation.tikz}}
\right)\circ\left(\interpxw{
\begin{tikzpicture}
	\begin{pgfonlayer}{nodelayer}
		\node [style=gn] (0) at (0, -0) {$\frac{\pi}{4}$};
		\node [style=none] (1) at (0, 0.5000001) {};
		\node [style=none] (2) at (0, -0.5000001) {};
	\end{pgfonlayer}
	\begin{pgfonlayer}{edgelayer}
		\draw (1.center) to (2.center);
	\end{pgfonlayer}
\end{tikzpicture}}
}\right)^k\circ \left(
\InputIfFileExists{gn-0-n-1-interpretation.tikz}{}{\input{./figures/gn-0-n-1-interpretation.tikz}}
\right)\]
\[
\InputIfFileExists{rn-kpi_4.tikz}{}{\input{./figures/rn-kpi_4.tikz}}
\mapsto \interpxw{\left(
\begin{tikzpicture}
	\begin{pgfonlayer}{nodelayer}
		\node [style={H box}] (0) at (0, 0) {};
		\node [style=none] (1) at (0, 0.5) {};
		\node [style=none] (2) at (0, -0.5) {};
	\end{pgfonlayer}
	\begin{pgfonlayer}{edgelayer}
		\draw (2.center) to (1.center);
	\end{pgfonlayer}
\end{tikzpicture}
}
\right)^{\otimes m}}\circ\interpxw{
\InputIfFileExists{gn-kpi_4.tikz}{}{\input{./figures/gn-kpi_4.tikz}}
} \circ \interpxw{\left(
\begin{tikzpicture}
	\begin{pgfonlayer}{nodelayer}
		\node [style={H box}] (0) at (0, 0) {};
		\node [style=none] (1) at (0, 0.5) {};
		\node [style=none] (2) at (0, -0.5) {};
	\end{pgfonlayer}
	\begin{pgfonlayer}{edgelayer}
		\draw (2.center) to (1.center);
	\end{pgfonlayer}
\end{tikzpicture}
}
\right)^{\otimes n}}\]
\rule{\columnwidth}{0.5pt}\\

The interpretation of the spacial composition $\otimes$ might seem a tad cryptical. It is in fact a way of putting ``side-by-side'' the interpretations of $D_1$ and $D_2$, while at the same time making them share the two control wires. We can see it as:
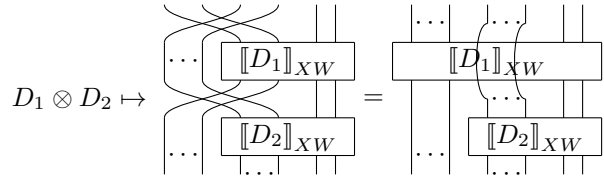
\[ D_1\otimes D_2 \mapsto 
\InputIfFileExists{interp-of-tensor-product-visualised.tikz}{}{\input{./figures/interp-of-tensor-product-visualised.tikz}}
 \]
It order for this to make sense, $D_1$ and $D_2$ should be able to commute on the control wires. This property is provided by the completeness of the \zwcalc, since it is semantically true.

One can check that $\interp{\interpxw{
\begin{tikzpicture}
	\begin{pgfonlayer}{nodelayer}
		\node [style={H box}] (0) at (0, 0) {};
		\node [style=none] (1) at (0, 0.5) {};
		\node [style=none] (2) at (0, -0.5) {};
	\end{pgfonlayer}
	\begin{pgfonlayer}{edgelayer}
		\draw (2.center) to (1.center);
	\end{pgfonlayer}
\end{tikzpicture}
}
}} = \psi\left( \interp{
\begin{tikzpicture}
	\begin{pgfonlayer}{nodelayer}
		\node [style={H box}] (0) at (0, 0) {};
		\node [style=none] (1) at (0, 0.5) {};
		\node [style=none] (2) at (0, -0.5) {};
	\end{pgfonlayer}
	\begin{pgfonlayer}{edgelayer}
		\draw (2.center) to (1.center);
	\end{pgfonlayer}
\end{tikzpicture}
}
} \right) = \frac{1}{2}\begin{pmatrix}1&1\\1&-1\end{pmatrix}\otimes (M-M^3)$ and $\interp{\interpxw{
\begin{tikzpicture}
	\begin{pgfonlayer}{nodelayer}
		\node [style=gn] (0) at (0, -0) {$\frac{\pi}{4}$};
		\node [style=none] (1) at (0, 0.5000001) {};
		\node [style=none] (2) at (0, -0.5000001) {};
	\end{pgfonlayer}
	\begin{pgfonlayer}{edgelayer}
		\draw (1.center) to (2.center);
	\end{pgfonlayer}
\end{tikzpicture}}
}} = \psi\left( \interp{
\begin{tikzpicture}
	\begin{pgfonlayer}{nodelayer}
		\node [style=gn] (0) at (0, -0) {$\frac{\pi}{4}$};
		\node [style=none] (1) at (0, 0.5000001) {};
		\node [style=none] (2) at (0, -0.5000001) {};
	\end{pgfonlayer}
	\begin{pgfonlayer}{edgelayer}
		\draw (1.center) to (2.center);
	\end{pgfonlayer}
\end{tikzpicture}}
} \right) = \begin{pmatrix}I_4 & 0 \\ 0 & M\end{pmatrix}$. More generally:

    \begin{proposition}
      \label{prop:pxwpsi}
      Let $D$ be a diagram of the \zxtcalc. Then
      \[\interp{\interpxw{D}}=\psi(\interp{D})\]

      In particular, if $\interp{D_1} = \interp{D_2}$ then
      $\interp{\interpxw{D_1}} = \interp{\interpxw{D_2}}$      
    \end{proposition}
    The proof is a straightforward induction using the fact
    that $\psi$ is an homomorphism. Slight care has to be taken to
    treat the case of $D_1 \otimes D_2$:\\
Suppose $\interp{D_1}=\sum\limits_{k=0}^3A_k e^{i\frac{k\pi}{4}}$ and $\interp{D_2}=\sum\limits_{k=0}^3B_k e^{i\frac{k\pi}{4}}$ are their \emph{unique} decomposition, and that $\interp{\interp{D_i}_{XW}}=\psi(\interp{D_i})$. Then:
\begin{align*}
\interp{\interp{D_1\otimes D_2}_{XW}} &= \left(I\otimes \psi(\interp{D_1})\right)\circ \interp{\scalebox{0.7}{\nmcrossI{}{}~}}\circ \left(I\otimes \sum\limits_{k=0}^3A_k\otimes M^k\right)\circ \interp{\scalebox{0.7}{\nmcrossI{}{}~}}\\
&= \left(\sum\limits_{l=0}^3I\otimes B_l \otimes M^l\right)\circ \left(\sum\limits_{k=0}^3A_k \otimes I\otimes M^k\right)
= \sum_{k,l}A_k\otimes B_l \otimes M^{k+l}\\
&= \psi\left(\sum_{k,l}(A_k\otimes B_l) e^{i\frac{(k+l)\pi}{4}}\right) = \psi(\interp{D_1\otimes D_2})
\end{align*}

\section{From \zwh to \zxt-Diagrams}
\label{section:ZWZX}

We define here an interpretation $\interpwx{.}$ that transforms any diagram of the \zwh-Calculus into a \zxt-diagram, which is easy to do since $\mathbb{D}\subset\mathbb{D}[\piq]$:\\
\titlerule{$\interpwx{.}$}
\vspace{-2em}
\begin{multicols}{3}
\[ 
\InputIfFileExists{empty-diagram.tikz}{}{\input{./figures/empty-diagram.tikz}}
 \quad\mapsto\quad 
\InputIfFileExists{empty-diagram.tikz}{}{\input{./figures/empty-diagram.tikz}}
\]\vfill
\[ 
\begin{tikzpicture}
	\begin{pgfonlayer}{nodelayer}
		\node [style=none] (0) at (0, 0.2499999) {};
		\node [style=none] (1) at (0, -0.2499999) {};
	\end{pgfonlayer}
	\begin{pgfonlayer}{edgelayer}
		\draw (0.center) to (1.center);
	\end{pgfonlayer}
\end{tikzpicture}}
 \quad\mapsto\quad 
\begin{tikzpicture}
	\begin{pgfonlayer}{nodelayer}
		\node [style=none] (0) at (0, 0.2499999) {};
		\node [style=none] (1) at (0, -0.2499999) {};
	\end{pgfonlayer}
	\begin{pgfonlayer}{edgelayer}
		\draw (0.center) to (1.center);
	\end{pgfonlayer}
\end{tikzpicture}}
\]\vfill
\[ 
\begin{tikzpicture}
	\begin{pgfonlayer}{nodelayer}
		\node [style=none] (0) at (-0.2500001, -0) {};
		\node [style=none] (1) at (0.2500001, -0) {};
	\end{pgfonlayer}
	\begin{pgfonlayer}{edgelayer}
		\draw [bend left=90, looseness=1.75] (0.center) to (1.center);
	\end{pgfonlayer}
\end{tikzpicture}}
 \quad \raisebox{0.3em}{$\mapsto$} \quad 
\begin{tikzpicture}
	\begin{pgfonlayer}{nodelayer}
		\node [style=none] (0) at (-0.2500001, -0) {};
		\node [style=none] (1) at (0.2500001, -0) {};
	\end{pgfonlayer}
	\begin{pgfonlayer}{edgelayer}
		\draw [bend left=90, looseness=1.75] (0.center) to (1.center);
	\end{pgfonlayer}
\end{tikzpicture}}
\]\vfill
\[ 
\begin{tikzpicture}
	\begin{pgfonlayer}{nodelayer}
		\node [style=none] (0) at (-0.2500001, 0.2500001) {};
		\node [style=none] (1) at (0.2500001, 0.2500001) {};
	\end{pgfonlayer}
	\begin{pgfonlayer}{edgelayer}
		\draw [bend right=90, looseness=1.75] (0.center) to (1.center);
	\end{pgfonlayer}
\end{tikzpicture}}
 \quad\raisebox{0.3em}{$\mapsto$}\quad 
\begin{tikzpicture}
	\begin{pgfonlayer}{nodelayer}
		\node [style=none] (0) at (-0.2500001, 0.2500001) {};
		\node [style=none] (1) at (0.2500001, 0.2500001) {};
	\end{pgfonlayer}
	\begin{pgfonlayer}{edgelayer}
		\draw [bend right=90, looseness=1.75] (0.center) to (1.center);
	\end{pgfonlayer}
\end{tikzpicture}}
\]\vfill
\[
\InputIfFileExists{ZW-to-ZX-braid-no-braid.tikz}{}{\input{./figures/ZW-to-ZX-braid-no-braid.tikz}}
\]\vfill
\[
\InputIfFileExists{ZW-to-ZX-cross-no-braid.tikz}{}{\input{./figures/ZW-to-ZX-cross-no-braid.tikz}}
\]\vfill
\[ \half \quad\mapsto\quad 
\InputIfFileExists{half-ZX.tikz}{}{\input{./figures/half-ZX.tikz}}
 \]\vfill
\[
\InputIfFileExists{ZW-to-ZX-white-dot-1-1.tikz}{}{\input{./figures/ZW-to-ZX-white-dot-1-1.tikz}}
\]\vfill
\[
\InputIfFileExists{ZW-to-ZX-white-dot-2-1.tikz}{}{\input{./figures/ZW-to-ZX-white-dot-2-1.tikz}}
\]\vfill
\[
\InputIfFileExists{ZW-to-ZX-dot-1-1.tikz}{}{\input{./figures/ZW-to-ZX-dot-1-1.tikz}}
\]
\[
\InputIfFileExists{ZW-to-ZX-dot-1-2.tikz}{}{\input{./figures/ZW-to-ZX-dot-1-2.tikz}}
\]
\end{multicols}

\noindent\begin{minipage}{\columnwidth}
\[D_1\circ D_2\mapsto \interpwx{D_1}\circ\interpwx{D_2}\qquad D_1\otimes D_2\mapsto \interpwx{D_1}\otimes\interpwx{D_2}\]
\rule{\columnwidth}{0.5pt}
\end{minipage}\\

\begin{proposition}
  \label{prop:interpwx}
  Let $D$ be a diagram of the \zwh calculus.
  Then $\interp{\interpwx{D}}=\interp{D}$
\end{proposition}  
The proof is by induction on $D$.

This interpretation $\interp{.}_{WX}$ from the \zwcalc to the \zxcalc is pretty straightforward, except for the three-legged black node. This is where some syntactic sugar can come in handy.
\begin{definition}
\label{def:triangle}
We define the ``triangle node'' as:
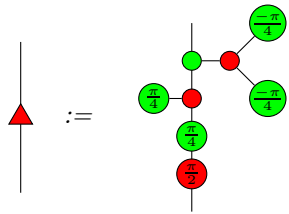
\[
\InputIfFileExists{ug-decomp.tikz}{}{\input{./figures/ug-decomp.tikz}}
\]
\end{definition}
One can check that $\interp{
\InputIfFileExists{ug-node.tikz}{}{\input{./figures/ug-node.tikz}}
} = \begin{pmatrix}1&1\\0&1\end{pmatrix}$. Then the interpretation of the three-legged black dot is simplified:
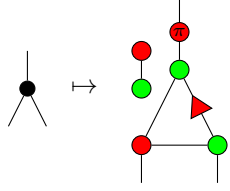
\[
\InputIfFileExists{ZW-to-ZX-dot-1-2-simplified.tikz}{}{\input{./figures/ZW-to-ZX-dot-1-2-simplified.tikz}}
\]
as is the rule \nfi (see Lemma \ref{lem:not-ug-is-symmetrical}):
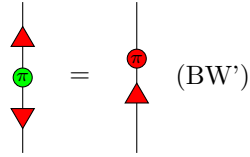
\[
\InputIfFileExists{not-ug-is-symmetrical.tikz}{}{\input{./figures/not-ug-is-symmetrical.tikz}}
\quad \textnormal{(BW')}\]

This shortcut will be very useful in the technical proof of the completeness of the language for Clifford+T with the set of rules in Figure \ref{fig:ZX_rules}.

\begin{proposition}
\label{prop:rules-preserved}
The interpretation $\interpwx{.}$ preserves all the rules of the \zwh-Calculus:
\[ \zwh\vdash D_1=D_2 \quad\implies\quad \zxt\vdash\interpwx{D_1}=\interpwx{D_2} \]
\end{proposition}
The proof is in appendix at Section \ref{prf:rules-preserved}.

\section{Completeness of the \frag{4} of the \zxcalc}
\label{section:complete}

To finish the proof it remains to compose the two interpretations:

\begin{proposition}
\label{prop:left-inverse-ZX}
We can retrieve any ZX$_{\pi/4}$-diagram from its image under the composition of the two interpretations:
\begin{align*}
\forall D\in \zxt,\quad
\zxt\vdash D =\left(\hspace{-0.2em}\scalebox{0.8}{
\InputIfFileExists{bottom-composition.tikz}{}{\input{./figures/bottom-composition.tikz}}
}\hspace{-0.3em}\right)\circ\interpwx{\interpxw{D}}\circ\left(\hspace{-0.2em}\scalebox{0.8}{
\InputIfFileExists{top-composition.tikz}{}{\input{./figures/top-composition.tikz}}
}\hspace{-0.3em}\right)
\end{align*}
\end{proposition}
The proof is in appendix at Section \ref{prf:left-inverse-ZX}.

\begin{corollary}
\label{cor:composition}
  If $\zxt \vdash \interpwx{\interpxw{D_1}} = \interpwx{\interpxw{D_2}}$
  then $\zxt \vdash D_1 = D_2$.
\end{corollary}  

%

\section{Expressive power of the \zxt-diagrams}
\label{section:exp}

The \zwcalc is complete, and additionally any integer matrix can be represented in the \zwcalc \cite{zw}.
A similar result follows immediately for the \zwh-calculus.

\begin{proposition}
\zwh-Diagrams are universal for matrices of $\mathbb{D}^{2^n}\times\mathbb{D}^{2^m}$:
\[\forall A\in \mathbb{D}^{2^n}\times\mathbb{D}^{2^m},~~\exists D\in \zwh,~~ \interp{D}=A\]
\end{proposition}

Regarding the expressive power of \zxt-diagrams, since the unitary matrices over $\mathbb{D}[\piq]$ are representable with Clifford+T circuits \cite{SelingerGiles}, so are they with \zxt-diagrams. We actually show that any matrix over $\mathbb{D}[\piq]$ can be represented by a \zxt-diagram:

\begin{proposition}
\label{prop:universality}
The \frag{4} of the \zxcalc represents exactly  matrices over $\mathbb{D}[e^{i\frac{\pi}{4}}]$:
\[\forall A\in \mathbb{D}[e^{i\frac{\pi}{4}}]^{2^n\times 2^m},~~\exists D\in \zxt,~~ \interp{D}=A\]
\end{proposition}
\begin{proof}
Let $A\in \mathbb{D}[e^{i\frac{\pi}{4}}]^{2^n\times 2^m}$. We define $A'=\psi(A)\in\mathbb{D}^{2^{n+2}\times 2^{m+2}}$. Since \zwh-diagrams are universal for matrices over dyadic rationals: $\exists D\in \zwh,~\interp{D}=A'$. Since $\interpwx{.}$ preserves the semantics, we can define a ZX-diagram of the \frag{4} $D'=\interpwx{D}$ such that $\interp{D'}=A'$.\\
Now, notice that $\theta = \begin{pmatrix}1\\\piq\\(\piq)^2\\(\piq)^3\end{pmatrix} =\interp{\scalebox{0.8}{
\InputIfFileExists{theta.tikz}{}{\input{./figures/theta.tikz}}
}}$
,\\and $e_1=\begin{pmatrix}1&0&0&0\end{pmatrix} = \interp{\scalebox{0.8}{
\InputIfFileExists{projector-1-0-0-0.tikz}{}{\input{./figures/projector-1-0-0-0.tikz}}
}}$, so if we apply the second diagram at the two bottom right wires, and the first state on the two top right wires of $D'$, we end up with $D''$ such that $\interp{D''}=A$. Indeed:\\
$\interp{D''} = (I\otimes
e_1)\circ\interp{D'}\circ\left(I\otimes\theta\right)
= (I\otimes e_1)\circ\psi(A)\circ\left(I\otimes\theta\right)
=(I\otimes e_1)\circ(A\otimes\theta)
=A\otimes\left(e_1\circ\theta\right)
= A  $
\qed\end{proof}

\section{Discussion on the New Rules}
\label{sec:discussion}

\label{sec:dicussion}
We can try and give an explanation of the rule \com, in terms of commutation of controlled operations. 
Consider the following diagram:
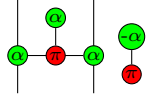
\[
\InputIfFileExists{control-2-alpha.tikz}{}{\input{./figures/control-2-alpha.tikz}}
\]
It is a \emph{controlled operation}. Indeed, if \scalebox{0.8}{
\begin{tikzpicture}
	\begin{pgfonlayer}{nodelayer}
		\node [style=rn] (0) at (0, 0.25) {};
		\node [style=none] (1) at (0, -0.25) {};
	\end{pgfonlayer}
	\begin{pgfonlayer}{edgelayer}
		\draw [style=none] (0) to (1.center);
	\end{pgfonlayer}
\end{tikzpicture}}
} is plugged on the left wire, we obtain the identity on the right one:
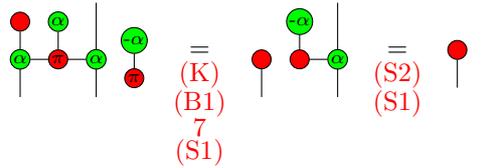
\def\fig{control-2-alpha-0}
\begin{align*}
\input{./figures/\fig/\fig_00.tikz}
\eq{\kt\\\bo\\\ref{lem:bicolor-0-alpha}\\\so}\input{./figures/\fig/\fig_01.tikz}
\eq{\st\\\so}\input{./figures/\fig/\fig_02.tikz}
\end{align*}
and if \scalebox{0.8}{
\begin{tikzpicture}
	\begin{pgfonlayer}{nodelayer}
		\node [style=rn] (0) at (0, 0.25) {$\pi$};
		\node [style=none] (1) at (0, -0.25) {};
	\end{pgfonlayer}
	\begin{pgfonlayer}{edgelayer}
		\draw [style=none] (0) to (1.center);
	\end{pgfonlayer}
\end{tikzpicture}}
} is plugged:
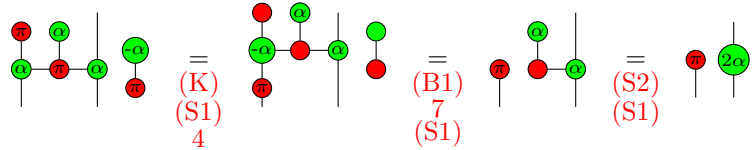
\def\fig{control-2-alpha-1}
\begin{align*}
\input{./figures/\fig/\fig_00.tikz}
\eq{\kt\\\so\\\ref{lem:k1}}\input{./figures/\fig/\fig_01.tikz}
\eq{\bo\\\ref{lem:bicolor-0-alpha}\\\so}\input{./figures/\fig/\fig_02.tikz}
\eq{\st\\\so}\input{./figures/\fig/\fig_03.tikz}
\end{align*}
Hence, the diagram is actually the controlled $Z$-rotation $\Lambda R_Z(2\alpha)$. Similarly, one can show that the following diagram represents the controlled $X$-rotation $\Lambda R_X(2\alpha)$:
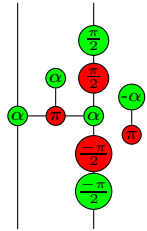
\[
\InputIfFileExists{control-X-alpha.tikz}{}{\input{./figures/control-X-alpha.tikz}}
\]

A controlled operation $\Lambda U$ can be turned into an  \emph{``anti-controlled'' operation} $\overline{\Lambda}U$ (where the role of \scalebox{0.8}{
\begin{tikzpicture}
	\begin{pgfonlayer}{nodelayer}
		\node [style=rn] (0) at (0, 0.25) {};
		\node [style=none] (1) at (0, -0.25) {};
	\end{pgfonlayer}
	\begin{pgfonlayer}{edgelayer}
		\draw [style=none] (0) to (1.center);
	\end{pgfonlayer}
\end{tikzpicture}}
} and \scalebox{0.8}{
\begin{tikzpicture}
	\begin{pgfonlayer}{nodelayer}
		\node [style=rn] (0) at (0, 0.25) {$\pi$};
		\node [style=none] (1) at (0, -0.25) {};
	\end{pgfonlayer}
	\begin{pgfonlayer}{edgelayer}
		\draw [style=none] (0) to (1.center);
	\end{pgfonlayer}
\end{tikzpicture}}
} are exchanged) by adding two 
 red $\pi$-nodes on the first wire, one on the input and one on the output. 
 For instance, $\overline \Lambda R_Z(2\beta)$ is expressed as:
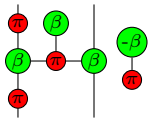
\[
\InputIfFileExists{anti-control-beta.tikz}{}{\input{./figures/anti-control-beta.tikz}}
\]

Controlled and anti-controlled operations obviously commute: for any operations $U$ and $V$ of the same size, 
$\Lambda U \circ \overline \Lambda V = \overline \Lambda V \circ \Lambda U$. To derive this result in the ZX-calculus for $\Lambda R_X(2\alpha)$ and $\overline \Lambda R_Z(2\beta)$ the use of rule \com seems to be required:
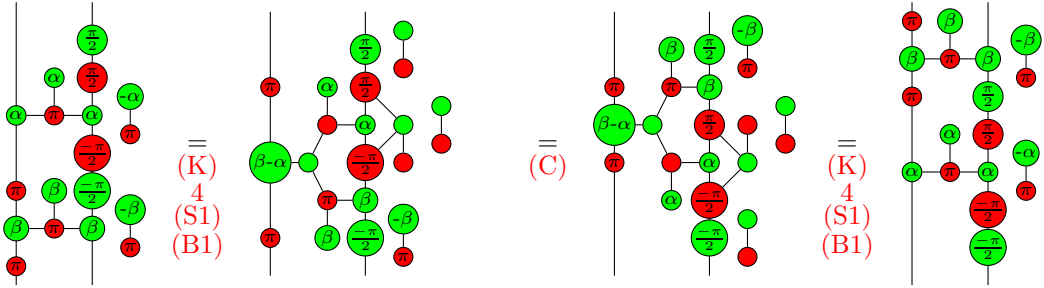
\def\fig{control-X-alpha-anti-control-Z-beta-2}
\begin{align*}
\input{./figures/\fig/\fig_00.tikz}
&\eq{\kt\\\ref{lem:k1}\\\so\\\bo}\input{./figures/\fig/\fig_01.tikz}
&\eq{\com}\input{./figures/\fig/\fig_02.tikz}
\eq{\kt\\\ref{lem:k1}\\\so\\\bo}\input{./figures/\fig/\fig_03.tikz}
\end{align*}

The rule \com itself actually goes a step further and directly expresses the equality $\Lambda U(\alpha,\gamma) \circ \overline \Lambda (R_Z(2\beta)\otimes \mathbb{I}) = \overline \Lambda (R_Z(2\beta)\otimes \mathbb{I}) \circ \Lambda U(\alpha,\gamma)$ for
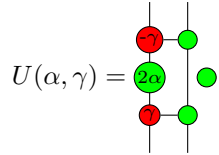
\[U(\alpha,\gamma) = 
\InputIfFileExists{diagram-U.tikz}{}{\input{./figures/diagram-U.tikz}}
\]
$\Lambda U(\alpha,\gamma)$ is a control of a $2$-qubits operation, such as the Toffoli gate, which controls CNOT. Because the operation it controls (CNOT) is also a controlled operation, it can also be seen as an operation on one wire, controlled by two others.

Since the Toffoli gate is represented by a matrix over $\mathbb{Z}$, it is possible to express it in the \frag4 of the ZX-Calculus. Indeed, using the \scalebox{0.8}{
\InputIfFileExists{ug-node.tikz}{}{\input{./figures/ug-node.tikz}}
}-notation:
\def\fig{toffoli-from-triangle}
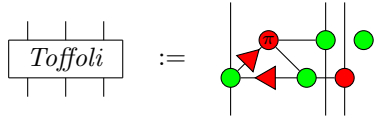
\[\input{./figures/\fig/\fig_00.tikz}\quad:=\quad\input{./figures/\fig/\fig_01.tikz}\]
One can check that this diagram actually represents the Toffoli gate, for instance by deriving, for any $k,\ell \in \{0,1\}$: 
\def\fig{toffoli-red-state}
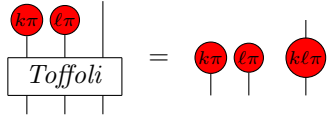
\begin{align*}
\input{./figures/\fig/\fig_00.tikz}\eq{}\input{./figures/\fig/\fig_01.tikz}
\end{align*}
More surprisingly, by plugging particular states, operation and projectors on the Toffoli gate, one can recover the triangle node:
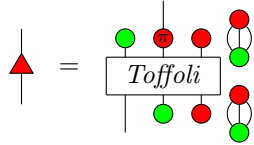
\def\fig{triangle-from-toffoli}
\[\input{./figures/\fig/\fig_00.tikz}\eq{}\input{./figures/\fig/\fig_01.tikz}\]
This shows a connection between the diagram denoted \scalebox{0.8}{
\InputIfFileExists{ug-node.tikz}{}{\input{./figures/ug-node.tikz}}
} and the Toffoli gate. As pointed out in section \ref{section:ZWZX}, the rule \nfi can be greatly simplified when using the \scalebox{0.8}{
\InputIfFileExists{ug-node.tikz}{}{\input{./figures/ug-node.tikz}}
}-notation.
Hence, \nfi might help to find rules for a potential axiomatisation of the quantum circuits, or in fact any graphical language that contains the Toffoli gate.

\vspace{0.2cm}

We leave open the question of minimality of the axiomatisation in Figure \ref{fig:ZX_rules}. Many of the previous axioms were proven to be necessary (not derivable from the others) \cite{towards-minimal,supplementarity}, but some proofs may not hold with the addition of the two new axioms, and there is currently no known proof of the necessity of \com and \nfi.
\section*{Acknowledgements}                          
The authors acknowledge support from the projects ANR-17-CE25-0009 SoftQPro, ANR-17-CE24-0035 VanQuTe, PIA-GDN/Quantex, and STIC-AmSud 16-STIC-05 FoQCoSS. All diagrams were written with the help of TikZit.


\appendix
\section{Appendix}

In this appendix (\ref{prf:rules-preserved}, \ref{prf:left-inverse-ZX}) are the proofs of Propositions \ref{prop:rules-preserved} and \ref{prop:left-inverse-ZX}. To simplify the following work, we use the new node introduced as a notation in Section \ref{section:ZWZX}, and give a few lemmas in Section \ref{app:lemmas}, and prove them in Section \ref{app:proof-of-lemmas}. Keep in mind that for any provable equation, its upside down version, its colour-swapped version, and (after Lemma \ref{lem:-1-interp}) its version with opposed angles are all provable.

\subsection{Lemmas}
\label{app:lemmas}
\begin{multicols}{3}
\begin{lemma}
\label{lem:2-is-sqrt2-squared}
\[
\InputIfFileExists{2-is-sqrt2-squared.tikz}{}{\input{./figures/2-is-sqrt2-squared.tikz}}
\]
\end{lemma}

\begin{lemma}
\label{lem:hopf}
\[
\InputIfFileExists{hopf.tikz}{}{\input{./figures/hopf.tikz}}
\]
\end{lemma}

\begin{lemma}
\label{lem:k1}
\[
\InputIfFileExists{k1.tikz}{}{\input{./figures/k1.tikz}}
\]
\end{lemma}

\begin{lemma}
\label{lem:inverse}
\[
\InputIfFileExists{inverse.tikz}{}{\input{./figures/inverse.tikz}}
\]
\end{lemma}

\begin{lemma}
\label{lem:multiplying-global-phases}
\[
\InputIfFileExists{multiplying-global-phases.tikz}{}{\input{./figures/multiplying-global-phases.tikz}}
\]
\end{lemma}

\begin{lemma}
\label{lem:bicolor-0-alpha}
\[
\InputIfFileExists{bicolor-0-alpha.tikz}{}{\input{./figures/bicolor-0-alpha.tikz}}
\]
\end{lemma}

\begin{lemma}
\label{lem:hadamard-involution}
\[
\InputIfFileExists{hadamard-involution.tikz}{}{\input{./figures/hadamard-involution.tikz}}
\]
\end{lemma}

\begin{lemma}
\label{lem:h-loop}
\[
\InputIfFileExists{h-loop.tikz}{}{\input{./figures/h-loop.tikz}}
\]
\end{lemma}

\begin{lemma}
\label{lem:red-pi_2-around-green-node}
\[
\InputIfFileExists{lemma-red-pi_2-around-green-node.tikz}{}{\input{./figures/lemma-red-pi_2-around-green-node.tikz}}
\]
\end{lemma}

\begin{lemma}
\label{lem:6-cycle-bialgebra}
\[{
\InputIfFileExists{6-cycle-bialgebra.tikz}{}{\input{./figures/6-cycle-bialgebra.tikz}}
}\]
\end{lemma}

\begin{lemma}
\label{lem:control-pi-and-anti-CNOT-commute}
\[
\InputIfFileExists{control-pi-and-anti-CNOT-commute.tikz}{}{\input{./figures/control-pi-and-anti-CNOT-commute.tikz}}
\]
\end{lemma}

\end{multicols}

\begin{lemma}
\label{lem:-1-interp}
Let $\interp{.}_{-1}$ be the interpretation that multiplies all the angles by $-1$. Then:
\[ \zx\vdash D_1=D_2 \iff \zx\vdash \interp{D_1}_{-1}=\interp{D_2}_{-1} \]
\end{lemma}

\begin{multicols}{3}

\begin{lemma}
\label{lem:gn-pi_2-0-0-equals-sqrt2-exp-pi_4}
\[
\begin{tikzpicture}
	\begin{pgfonlayer}{nodelayer}
		\node [style=gn] (0) at (-1, -0) {$\frac{\pi}{2}$};
		\node [style=none] (1) at (0, -0) {=};
		\node [style=gn] (2) at (1, -0.2500001) {$\frac{\pi}{4}$};
		\node [style=rn] (3) at (1, 0.5) {$\pi$};
	\end{pgfonlayer}
	\begin{pgfonlayer}{edgelayer}
		\draw (3) to (2);
	\end{pgfonlayer}
\end{tikzpicture}}
\]
\end{lemma}

\begin{lemma}
\label{lem:green-state-pi_2-is-red-state-minus-pi_2}
\[
\InputIfFileExists{green-state-pi_2-is-red-state-minus-pi_2.tikz}{}{\input{./figures/green-state-pi_2-is-red-state-minus-pi_2.tikz}}
\]
\end{lemma}

\begin{lemma}
\label{lem:euler-decomp-with-scalar}
\[
\InputIfFileExists{euler-decomp-with-scalar.tikz}{}{\input{./figures/euler-decomp-with-scalar.tikz}}
\]
\end{lemma}

\begin{lemma}
\label{lem:C1-original}
\[
\InputIfFileExists{control-commutation-2.tikz}{}{\input{./figures/control-commutation-2.tikz}}
\]
\end{lemma}

\begin{lemma}
\label{lem:C1-bis}
\[
\InputIfFileExists{control-commutation-3.tikz}{}{\input{./figures/control-commutation-3.tikz}}
\]
\end{lemma}


\begin{lemma}
\label{lem:supp-to-minus-pi_4}
\[
\InputIfFileExists{supp-to-minus-pi_4.tikz}{}{\input{./figures/supp-to-minus-pi_4.tikz}}
\]
\end{lemma}

\begin{lemma}
\label{lem:triangle-decomp-2}
\[
\InputIfFileExists{triangle-decomp-2.tikz}{}{\input{./figures/triangle-decomp-2.tikz}}
\]
\end{lemma}

\begin{lemma}
\label{lem:not-triangle-is-symmetrical}
\[
\InputIfFileExists{not-ug-is-symmetrical-2.tikz}{}{\input{./figures/not-ug-is-symmetrical-2.tikz}}
\]
\end{lemma}

\begin{lemma}
\label{lem:red-state-on-triangle}
\[
\InputIfFileExists{red-state-on-triangle.tikz}{}{\input{./figures/red-state-on-triangle.tikz}}
\]
\end{lemma}

\begin{lemma}
\label{lem:pi-red-state-on-triangle}
\[
\InputIfFileExists{pi-red-state-on-triangle.tikz}{}{\input{./figures/pi-red-state-on-triangle.tikz}}
\]
\end{lemma}

\begin{lemma}
\label{lem:red-state-on-upside-down-triangle}
\[
\InputIfFileExists{red-state-on-upside-down-triangle.tikz}{}{\input{./figures/red-state-on-upside-down-triangle.tikz}}
\]
\end{lemma}

\begin{lemma}
\label{lem:pi-red-state-on-upside-down-triangle}
\[
\InputIfFileExists{pi-red-state-on-upside-down-triangle.tikz}{}{\input{./figures/pi-red-state-on-upside-down-triangle.tikz}}
\]
\end{lemma}

\begin{lemma}
\label{lem:pi-green-state-on-upside-down-triangle}
\[
\InputIfFileExists{pi-green-state-on-upside-down-triangle.tikz}{}{\input{./figures/pi-green-state-on-upside-down-triangle.tikz}}
\]
\end{lemma}

\begin{lemma}
\label{lem:pi-green-state-on-triangle}
\[
\InputIfFileExists{pi-green-state-on-triangle.tikz}{}{\input{./figures/pi-green-state-on-triangle.tikz}}
\]
\end{lemma}

\begin{lemma}
\label{lem:triangle-hadamard-parallel}
\[
\InputIfFileExists{anti-control-pi-control-green.tikz}{}{\input{./figures/anti-control-pi-control-green.tikz}}
\]
\end{lemma}

\begin{lemma}
\label{lem:looped-triangle}
\[
\InputIfFileExists{looped-triangle.tikz}{}{\input{./figures/looped-triangle.tikz}}
\]
\end{lemma}

\begin{lemma}
\label{lem:black-dot-swappable-outputs}
\[
\InputIfFileExists{lemma-black-dot-swappable-outputs.tikz}{}{\input{./figures/lemma-black-dot-swappable-outputs.tikz}}
\]
\end{lemma}

\begin{lemma}
\label{lem:not-ug-is-symmetrical}
\[
\InputIfFileExists{not-ug-is-symmetrical.tikz}{}{\input{./figures/not-ug-is-symmetrical.tikz}}
\]
\end{lemma}

\begin{lemma}
\label{lem:inverse-of-triangle}
\[
\InputIfFileExists{inverse-of-triangle.tikz}{}{\input{./figures/inverse-of-triangle.tikz}}
\]
\end{lemma}

\begin{lemma}
\label{lem:symmetric-diagram-with-triangle-hadamard}
\[
\InputIfFileExists{lemma-symmetric-diagram-with-triangle-hadamard.tikz}{}{\input{./figures/lemma-symmetric-diagram-with-triangle-hadamard.tikz}}
\]
\end{lemma}

\begin{lemma}
\label{lem:anti-control-hanging-branch-and-control-alpha-commute}
\[\fit{
\InputIfFileExists{anti-control-hanging-branch-and-control-alpha-commute.tikz}{}{\input{./figures/anti-control-hanging-branch-and-control-alpha-commute.tikz}}
}\]
\end{lemma}

\begin{lemma}
\label{lem:n-four}
\[
\InputIfFileExists{ug-and-not-W-commute-simplified.tikz}{}{\input{./figures/ug-and-not-W-commute-simplified.tikz}}
\]
\end{lemma}

\begin{lemma}
\label{lem:parallel-triangles}
\[
\InputIfFileExists{parallel-ugs-are-projections.tikz}{}{\input{./figures/parallel-ugs-are-projections.tikz}}
\]
\end{lemma}

\begin{lemma}
\label{lem:triangles-fork-absorbs-anti-CNOT}
\[
\InputIfFileExists{ug-fork-absorbs-anti-CNOT.tikz}{}{\input{./figures/ug-fork-absorbs-anti-CNOT.tikz}}
\] and \[
\InputIfFileExists{ug-fork-absorbs-CNOT.tikz}{}{\input{./figures/ug-fork-absorbs-CNOT.tikz}}
\]
\end{lemma}

\begin{lemma}
\label{lem:n-five}
\[
\InputIfFileExists{2-diagrams-of-control-triangle-axiom-simplified.tikz}{}{\input{./figures/2-diagrams-of-control-triangle-axiom-simplified.tikz}}
\]
\end{lemma}
\vfill\null
\end{multicols}

\subsection{Proof of Lemmas}
\label{app:proof-of-lemmas}

\begin{proof}[Lemmas \ref{lem:2-is-sqrt2-squared} 
 to \ref{lem:control-pi-and-anti-CNOT-commute}]$~$\\
Lemmas \ref{lem:multiplying-global-phases} and \ref{lem:bicolor-0-alpha} are proven in \cite{simplified-stabilizer,gen-supp}. The other lemmas are in the \frag2 and hence are derivable by completeness of this fragment.
\qed\end{proof}

\begin{proof}[Lemma \ref{lem:-1-interp}]
The result is quite obvious for all rules except maybe for \e, \eu and \nfi.
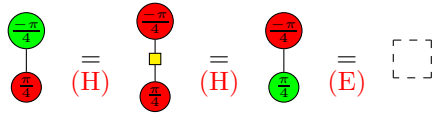
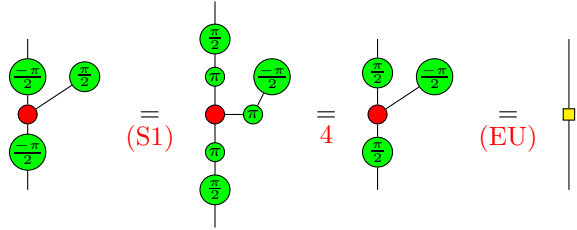
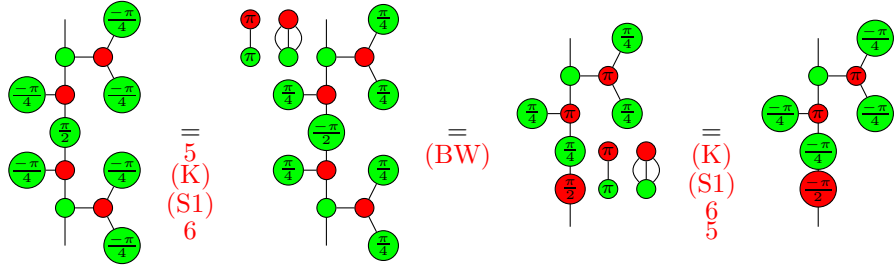
\begin{itemize}
\item \e:
\def\fig{bicolor_pi_4_eq_empty-minus-1}
\begin{align*}
\input{./figures/\fig/\fig_00.tikz}\eq{\h}\input{./figures/\fig/\fig_01.tikz}\eq{\h}
\input{./figures/\fig/\fig_02.tikz}\eq{\e}\input{./figures/\fig/\fig_03.tikz}
\end{align*}
\item \eu:
\def\fig{euler-decomp-scalar-free-minus-1}
\begin{align*}
\input{./figures/\fig/\fig_00.tikz}\eq{\so}\input{./figures/\fig/\fig_01.tikz}\eq{\ref{lem:k1}}
\input{./figures/\fig/\fig_02.tikz}\eq{\eu}\input{./figures/\fig/\fig_03.tikz}
\end{align*}
\item \nfi:
\def\fig{2-diagrams-of-control-triangle-axiom-minus-1-from-simplified}
\begin{align*}
\input{./figures/\fig/\fig_00.tikz}
\eq{\ref{lem:inverse}\\\kt\\\so\\\ref{lem:multiplying-global-phases}}\input{./figures/\fig/\fig_01.tikz}
\eq{\nfi}\input{./figures/\fig/\fig_02.tikz}
\eq{\kt\\\so\\\ref{lem:multiplying-global-phases}\\\ref{lem:inverse}}\input{./figures/\fig/\fig_03.tikz}
\end{align*}
\end{itemize}
Moreover, it is to be noticed that $\interp{
\InputIfFileExists{ug-node.tikz}{}{\input{./figures/ug-node.tikz}}
}_{-1} = 
\InputIfFileExists{ug-node.tikz}{}{\input{./figures/ug-node.tikz}}
$.
\qed\end{proof}


\begin{proof}[Lemma \ref{lem:gn-pi_2-0-0-equals-sqrt2-exp-pi_4}]
\def\fig{gn-pi_2-0-0-equals-sqrt2-exp-pi_4-proof}
\begin{align*}
\input{./figures/\fig/\fig_00.tikz}
\eq[~]{\e}\input{./figures/\fig/\fig_01.tikz}
\eq{\h}\input{./figures/\fig/\fig_02.tikz}
\eq{\eu}\input{./figures/\fig/\fig_03.tikz}
\eq{\kt\\\ref{lem:multiplying-global-phases}}\input{./figures/\fig/\fig_04.tikz}
\eq{\supp\\\ref{lem:hopf}}\input{./figures/\fig/\fig_05.tikz}
\eq[~]{\ref{lem:multiplying-global-phases}\\\ref{lem:inverse}}\input{./figures/\fig/\fig_06.tikz}
\end{align*}
\qed\end{proof}

\begin{proof}[Lemma
  \ref{lem:green-state-pi_2-is-red-state-minus-pi_2}]
\def\fig{green-state-pi_2-is-red-state-minus-pi_2-proof}
\begin{align*}
\input{./figures/\fig/\fig_00.tikz}\eq{\h}\input{./figures/\fig/\fig_01.tikz}\eq{\eu}\input{./figures/\fig/\fig_02.tikz}\eq{\bo}\input{./figures/\fig/\fig_03.tikz}\eq{\ref{lem:gn-pi_2-0-0-equals-sqrt2-exp-pi_4}\\\ref{lem:multiplying-global-phases}}\input{./figures/\fig/\fig_04.tikz}
\end{align*}
\qed\end{proof}

\begin{proof}[Lemma \ref{lem:euler-decomp-with-scalar}]
\def\fig{euler-decomp-with-scalar-proof}
\begin{align*}
\input{./figures/\fig/\fig_00.tikz}
\eq[]{\eu}\input{./figures/\fig/\fig_01.tikz}
\eq[]{\h\\\ref{lem:inverse}\\\ref{lem:multiplying-global-phases}}\input{./figures/\fig/\fig_02.tikz}
\eq[]{\ref{lem:green-state-pi_2-is-red-state-minus-pi_2}}\input{./figures/\fig/\fig_03.tikz}
\eq[]{\h\\\so}\input{./figures/\fig/\fig_04.tikz}
\end{align*}
\qed\end{proof}

\begin{proof}[Lemma \ref{lem:C1-original}]
\def\fig{control-commutation-2-proof}
\begin{align*}
\input{./figures/\fig/\fig_00.tikz}
\eq{\ref{lem:euler-decomp-with-scalar}\\\so}\input{./figures/\fig/\fig_01.tikz}
\eq{\bo}\input{./figures/\fig/\fig_02.tikz}
\eq[~]{\com}\input{./figures/\fig/\fig_03.tikz}
\eq[~]{\bo}\input{./figures/\fig/\fig_04.tikz}
\eq{\so\\\ref{lem:euler-decomp-with-scalar}}\input{./figures/\fig/\fig_05.tikz}
\end{align*}
\qed\end{proof}

\begin{proof}[Lemma \ref{lem:C1-bis}]
By completeness of the \frag2:
\[{
\InputIfFileExists{lemma-for-equivalence-c1.tikz}{}{\input{./figures/lemma-for-equivalence-c1.tikz}}
}\]
Then:
\def\fig{control-commutation-3-proof}
\begin{align*}
\input{./figures/\fig/\fig_00.tikz}
\eq{}\input{./figures/\fig/\fig_01.tikz}
\eq{\ref{lem:C1-original}}\input{./figures/\fig/\fig_02.tikz}
\eq{\ref{lem:hadamard-involution}}\input{./figures/\fig/\fig_03.tikz}
\eq{\ref{lem:hopf}}\input{./figures/\fig/\fig_04.tikz}
\end{align*}
\qed\end{proof}


\begin{proof}[Lemma \ref{lem:supp-to-minus-pi_4}]
\def\fig{supp-to-minus-pi_4-proof}
\begin{align*}
\input{./figures/\fig/\fig_00.tikz}
\eq{\ref{lem:inverse}\\\kt\\\ref{lem:multiplying-global-phases}}\input{./figures/\fig/\fig_01.tikz}
\eq{\ref{lem:C1-bis}\\\ref{lem:inverse}}\input{./figures/\fig/\fig_02.tikz}
\eq{\ref{lem:green-state-pi_2-is-red-state-minus-pi_2}\\\ref{lem:multiplying-global-phases}}\input{./figures/\fig/\fig_03.tikz}
\eq{\ref{lem:euler-decomp-with-scalar}\\\so\\\st\\\ref{lem:multiplying-global-phases}}\input{./figures/\fig/\fig_04.tikz}\\
\eq{\so\\\st\\\ref{lem:inverse}\\\ref{lem:hopf}}\input{./figures/\fig/\fig_05.tikz}
\eq{\kt\\\so\\\ref{lem:multiplying-global-phases}}\input{./figures/\fig/\fig_06.tikz}
\eq{\ref{lem:2-is-sqrt2-squared}\\\ref{lem:inverse}}\input{./figures/\fig/\fig_07.tikz}
\end{align*}
\qed\end{proof}


\begin{proof}[Lemma \ref{lem:triangle-decomp-2}] 
\def\fig{triangle-decomp-2-proof}
\begin{align*}
\input{./figures/\fig/\fig_00.tikz}
\eq{\ref{def:triangle}\\\kt\\\ref{lem:multiplying-global-phases}}\input{./figures/\fig/\fig_01.tikz}
\eq{\ref{lem:hopf}\\\so}\input{./figures/\fig/\fig_02.tikz}
\eq{\ref{lem:hadamard-involution}\\\ref{lem:C1-original}}\input{./figures/\fig/\fig_03.tikz}
\eq{\h\\\ref{lem:euler-decomp-with-scalar}\\\ref{lem:multiplying-global-phases}}\input{./figures/\fig/\fig_04.tikz}\\
\eq{\ref{lem:red-pi_2-around-green-node}\\\ref{lem:hopf}}\input{./figures/\fig/\fig_05.tikz}
\eq{\ref{lem:euler-decomp-with-scalar}\\\ref{lem:multiplying-global-phases}}\input{./figures/\fig/\fig_06.tikz}
\eq{\h}\input{./figures/\fig/\fig_07.tikz}
\end{align*}
\qed\end{proof}

\begin{proof}[Lemma \ref{lem:not-triangle-is-symmetrical}]
\def\fig{not-triangle-is-symmetrical-proof}
\begin{align*}
\input{./figures/\fig/\fig_00.tikz}
\eq{\ref{lem:triangle-decomp-2}}\input{./figures/\fig/\fig_01.tikz}
\eq{\ref{lem:6-cycle-bialgebra}\\\so}\input{./figures/\fig/\fig_02.tikz}
\eq{\so}\input{./figures/\fig/\fig_03.tikz}
\eq{\ref{lem:triangle-decomp-2}}\input{./figures/\fig/\fig_04.tikz}
\end{align*}
\qed\end{proof}

\begin{proof}[Lemma \ref{lem:red-state-on-triangle}]
\def\fig{red-state-on-triangle-proof}
\begin{align*}
\input{./figures/\fig/\fig_00.tikz}
\eq{\ref{def:triangle}}\input{./figures/\fig/\fig_01.tikz}
\eq{\ref{lem:inverse}\\\bo\\\so}\input{./figures/\fig/\fig_02.tikz}
\eq{\ref{lem:gn-pi_2-0-0-equals-sqrt2-exp-pi_4}}\input{./figures/\fig/\fig_03.tikz}
\eq{\ref{lem:green-state-pi_2-is-red-state-minus-pi_2}\\\ref{lem:multiplying-global-phases}}\input{./figures/\fig/\fig_04.tikz}
\eq{\so}\input{./figures/\fig/\fig_05.tikz}
\end{align*}
\qed\end{proof}

\begin{proof}[Lemma \ref{lem:pi-red-state-on-triangle}]
\def\fig{pi-red-state-on-triangle-proof}
\begin{align*}
\input{./figures/\fig/\fig_00.tikz}
\eq{\ref{def:triangle}}\input{./figures/\fig/\fig_01.tikz}
\eq{\ref{lem:k1}\\\bo}\input{./figures/\fig/\fig_02.tikz}
\eq{\kt\\\ref{lem:multiplying-global-phases}}\input{./figures/\fig/\fig_03.tikz}
\eq{\so\\\bo\\\ref{lem:bicolor-0-alpha}}\input{./figures/\fig/\fig_04.tikz}
\eq{\ref{lem:2-is-sqrt2-squared}\\\ref{lem:inverse}}\input{./figures/\fig/\fig_05.tikz}
\end{align*}
\qed\end{proof}

\begin{proof}[Lemmas \ref{lem:red-state-on-upside-down-triangle} and \ref{lem:pi-red-state-on-upside-down-triangle}]
The result comes naturally from \ref{lem:not-triangle-is-symmetrical}, \ref{lem:pi-red-state-on-triangle} and \ref{lem:red-state-on-triangle}.
\qed\end{proof}

\begin{proof}[Lemma \ref{lem:pi-green-state-on-upside-down-triangle}]
\def\fig{pi-green-state-on-upside-down-triangle-proof}
\begin{align*}
\input{./figures/\fig/\fig_00.tikz}
\eq[~]{\ref{def:triangle}}\input{./figures/\fig/\fig_01.tikz}
\eq[~]{\so\\\ref{lem:k1}\\\bo}\input{./figures/\fig/\fig_02.tikz}
\eq[~]{\supp\\\ref{lem:inverse}\\\ref{lem:hopf}}\input{./figures/\fig/\fig_03.tikz}
\eq[~]{\bo\\\so}\input{./figures/\fig/\fig_04.tikz}
\eq[~]{\ref{lem:gn-pi_2-0-0-equals-sqrt2-exp-pi_4}}\input{./figures/\fig/\fig_05.tikz}
\eq[~]{\ref{lem:multiplying-global-phases}\\\ref{lem:inverse}}\input{./figures/\fig/\fig_06.tikz}
\end{align*}
\qed\end{proof}

\begin{proof}[Lemma \ref{lem:pi-green-state-on-triangle}]
\def\fig{pi-green-state-on-triangle-proof}
\begin{align*}
\input{./figures/\fig/\fig_00.tikz}
\eq{\so\\\ref{lem:not-triangle-is-symmetrical}}\input{./figures/\fig/\fig_01.tikz}
\eq{\so\\\ref{lem:k1}\\\bo}\input{./figures/\fig/\fig_02.tikz}
\eq{\ref{lem:inverse}\\\ref{lem:pi-green-state-on-upside-down-triangle}}\input{./figures/\fig/\fig_03.tikz}
\eq{\so\\\ref{lem:k1}\\\bo\\\ref{lem:multiplying-global-phases}}\input{./figures/\fig/\fig_04.tikz}
\end{align*}
\qed\end{proof}

\begin{proof}[Lemma \ref{lem:triangle-hadamard-parallel}]
\def\fig{anti-control-pi-control-green-aux}
\begin{align*}
\input{./figures/\fig/\fig_00.tikz}
\eq{\h}\input{./figures/\fig/\fig_01.tikz}
\eq{\ref{lem:euler-decomp-with-scalar}\\\so\\\st\\\ref{lem:inverse}}\input{./figures/\fig/\fig_02.tikz}
\eq{\bt}\input{./figures/\fig/\fig_03.tikz}
\eq{\so\\\kt\\\ref{lem:multiplying-global-phases}}\input{./figures/\fig/\fig_04.tikz}
\eq{\ref{lem:euler-decomp-with-scalar}\\\so\\\st}\input{./figures/\fig/\fig_05.tikz}
\end{align*}
Then:
\def\fig{anti-control-pi-control-green-fin}
\begin{align*}
\input{./figures/\fig/\fig_00.tikz}
\eq{\ref{def:triangle}\\\so}\input{./figures/\fig/\fig_01.tikz}
\eq{}\input{./figures/\fig/\fig_02.tikz}
\eq{\ref{def:triangle}}\input{./figures/\fig/\fig_03.tikz}
\end{align*}
\qed\end{proof}

\begin{proof}[Lemma \ref{lem:looped-triangle}]
\def\fig{looped-triangle-proof}
\begin{align*}
\input{./figures/\fig/\fig_00.tikz}
\quad\eq[]{\ref{def:triangle}}\input{./figures/\fig/\fig_01.tikz}
\eq{\bt}\input{./figures/\fig/\fig_02.tikz}
\eq{\so}\input{./figures/\fig/\fig_03.tikz}
\eq{\ref{def:triangle}}\input{./figures/\fig/\fig_04.tikz}
\end{align*}
\qed\end{proof}

\begin{proof}[Lemma \ref{lem:black-dot-swappable-outputs}]
\def\fig{lemma-black-dot-swappable-outputs-proof}
\begin{align*}
\input{./figures/\fig/\fig_00.tikz}
\eq[\quad]{\bt}\input{./figures/\fig/\fig_01.tikz}
\eq[\quad]{\so\\\ref{lem:looped-triangle}}\input{./figures/\fig/\fig_02.tikz}
\end{align*}
\qed\end{proof}

\begin{proof}[Lemma \ref{lem:not-ug-is-symmetrical}]
\def\fig{not-ug-is-symmetrical-proof-from-simplified}
\begin{align*}
\input{./figures/\fig/\fig_00.tikz}
\eq{\ref{def:triangle}}\input{./figures/\fig/\fig_01.tikz}
\eq{\ref{lem:inverse}\\\kt\\\so}\input{./figures/\fig/\fig_02.tikz}
\eq{\so\\\st\\\kt\\\ref{lem:multiplying-global-phases}}\input{./figures/\fig/\fig_03.tikz}
\eq{\nfi}\input{./figures/\fig/\fig_04.tikz}\\
\eq{\ref{lem:k1}\\\kt\\\ref{lem:multiplying-global-phases}\\\ref{lem:inverse}}\input{./figures/\fig/\fig_05.tikz}
\eq{\ref{def:triangle}}\input{./figures/\fig/\fig_06.tikz}
\end{align*}
\qed\end{proof}

\begin{proof}[Lemma \ref{lem:inverse-of-triangle}]
\def\fig{inverse-of-triangle-proof}
\begin{align*}
\input{./figures/\fig/\fig_00.tikz}
\eq{\st\\\so\\\ref{lem:not-triangle-is-symmetrical}}\input{./figures/\fig/\fig_01.tikz}
\eq{\ref{def:triangle}\\\ref{lem:k1}\\\so}\input{./figures/\fig/\fig_02.tikz}
\eq{\ref{lem:supp-to-minus-pi_4}\\\ref{lem:k1}\\\ref{lem:inverse}\\\kt\\\ref{lem:multiplying-global-phases}}\input{./figures/\fig/\fig_03.tikz}
\eq{\supp\\\ref{lem:inverse}\\\ref{lem:hopf}}\input{./figures/\fig/\fig_04.tikz}
\eq{\ref{lem:gn-pi_2-0-0-equals-sqrt2-exp-pi_4}\\\kt\\\so}\input{./figures/\fig/\fig_05.tikz}
\eq{\so\\\st\\\ref{lem:multiplying-global-phases}\\\ref{lem:inverse}}\input{./figures/\fig/\fig_06.tikz}
\end{align*}
\qed\end{proof}

\begin{proof}[Lemma \ref{lem:symmetric-diagram-with-triangle-hadamard}]
First:
\def\fig{symmetric-diagram-with-triangle-hadamard-proof-1}
\begin{align*}
\input{./figures/\fig/\fig_00.tikz}
\eq{\ref{lem:euler-decomp-with-scalar}}\input{./figures/\fig/\fig_01.tikz}
\eq{\ref{lem:multiplying-global-phases}\\\ref{lem:gn-pi_2-0-0-equals-sqrt2-exp-pi_4}\\\so\\\st\\\supp}\input{./figures/\fig/\fig_02.tikz}
\eq{\st\\\kt\\\ref{lem:multiplying-global-phases}\\\ref{lem:k1}}\input{./figures/\fig/\fig_03.tikz}
\eq{\ref{lem:multiplying-global-phases}\\\ref{lem:supp-to-minus-pi_4}}\input{./figures/\fig/\fig_04.tikz}\\
\eq{\ref{def:triangle}\\\kt\\\ref{lem:multiplying-global-phases}}\input{./figures/\fig/\fig_05.tikz}
\eq{\ref{def:triangle}}\input{./figures/\fig/\fig_06.tikz}
\eq{\ref{lem:not-triangle-is-symmetrical}}\input{./figures/\fig/\fig_07.tikz}
\end{align*}
Then:
\def\fig{symmetric-diagram-with-triangle-hadamard-proof-2}
\begin{align*}
\input{./figures/\fig/\fig_00.tikz}
\eq{\ref{lem:inverse-of-triangle}}\input{./figures/\fig/\fig_01.tikz}
\eq{}\input{./figures/\fig/\fig_02.tikz}
\eq{\h}\input{./figures/\fig/\fig_03.tikz}
\eq{\ref{lem:not-triangle-is-symmetrical}}\input{./figures/\fig/\fig_04.tikz}
\end{align*}
\qed\end{proof}

\begin{proof}[Lemma \ref{lem:anti-control-hanging-branch-and-control-alpha-commute}]
First:
\def\fig{anti-control-hanging-branch-and-control-alpha-commute-proof-aux}
\begin{align*}
\input{./figures/\fig/\fig_00.tikz}
\eq{\h}\input{./figures/\fig/\fig_01.tikz}
\eq{\ref{lem:symmetric-diagram-with-triangle-hadamard}}\input{./figures/\fig/\fig_02.tikz}
\eq{\ref{def:triangle}\\\bt\\\ref{lem:k1}}\input{./figures/\fig/\fig_03.tikz}\\
\eq{\h}\input{./figures/\fig/\fig_04.tikz}
\eq{\ref{lem:C1-original}}\input{./figures/\fig/\fig_05.tikz}
\eq{\h\\\ref{lem:hopf}\\\so}\input{./figures/\fig/\fig_06.tikz}\\
\eq{\ref{lem:green-state-pi_2-is-red-state-minus-pi_2}}\input{./figures/\fig/\fig_07.tikz}
\eq{\ref{lem:supp-to-minus-pi_4}\\\st}\input{./figures/\fig/\fig_08.tikz}
\eq{\so\\\ref{lem:inverse}\\\ref{lem:hopf}\\\st}\input{./figures/\fig/\fig_09.tikz}
\eq{\h}\input{./figures/\fig/\fig_10.tikz}
\end{align*}
then:
\def\fig{anti-control-hanging-branch-and-control-alpha-commute-proof-fin}
\begin{align*}
\input{./figures/\fig/\fig_00.tikz}
\eq{\so\\\ref{lem:inverse}}\!\!\!\!\!\input{./figures/\fig/\fig_01.tikz}
\eq{\so\\\ref{lem:inverse}\\\ref{lem:hopf}\\\st\\\ref{def:triangle}}\input{./figures/\fig/\fig_02.tikz}
\eq{\bt\\\ref{lem:inverse}}\input{./figures/\fig/\fig_03.tikz}\\
\eq{\com}\input{./figures/\fig/\fig_04.tikz}
\eq{}\input{./figures/\fig/\fig_05.tikz}
\end{align*}
\qed\end{proof}

\begin{proof}[Lemma \ref{lem:n-four}]
\def\fig{ug-and-not-W-commute-simplified-proof-from-new-c1}
\begin{align*}
\input{./figures/\fig/\fig_00.tikz}
\eq{\st\\\so\\\ref{lem:not-triangle-is-symmetrical}}\input{./figures/\fig/\fig_01.tikz}
\eq{\ref{def:triangle}\\\ref{lem:k1}\\\so}\input{./figures/\fig/\fig_02.tikz}
\eq{\ref{lem:supp-to-minus-pi_4}\\\kt\\\ref{lem:multiplying-global-phases}}\input{./figures/\fig/\fig_03.tikz}
\eq{\bt\\\ref{lem:euler-decomp-with-scalar}}\input{./figures/\fig/\fig_04.tikz}\\
\eq{\h\\\so}\input{./figures/\fig/\fig_05.tikz}
\eq{\ref{lem:anti-control-hanging-branch-and-control-alpha-commute}\\\h}\input{./figures/\fig/\fig_06.tikz}
\eq{\ref{lem:euler-decomp-with-scalar}}\input{./figures/\fig/\fig_07.tikz}
\eq{\so\\\ref{lem:control-pi-and-anti-CNOT-commute}}\input{./figures/\fig/\fig_08.tikz}\\
\eq{\ref{lem:euler-decomp-with-scalar}}\input{./figures/\fig/\fig_09.tikz}
\eq{\h\\\kt\\\ref{lem:multiplying-global-phases}}\input{./figures/\fig/\fig_10.tikz}
\eq{\ref{lem:euler-decomp-with-scalar}\\\ref{lem:multiplying-global-phases}}\input{./figures/\fig/\fig_11.tikz}
\eq{\ref{lem:supp-to-minus-pi_4}\\\ref{def:triangle}\\\ref{lem:not-triangle-is-symmetrical}}\input{./figures/\fig/\fig_12.tikz}
\end{align*}
\qed\end{proof}


\begin{proof}[Lemma \ref{lem:parallel-triangles}]
First:
\def\fig{parallel-triangles-are-projections-proof}
\begin{align*}
\input{./figures/\fig/\fig_00.tikz}
\eq{\ref{lem:red-pi_2-around-green-node}}\input{./figures/\fig/\fig_01.tikz}
\eq{\so\\\h}\input{./figures/\fig/\fig_02.tikz}
\eq{\ref{lem:C1-original}\\\ref{lem:hadamard-involution}}\input{./figures/\fig/\fig_03.tikz}
\eq{\eu\\\so\\\ref{lem:green-state-pi_2-is-red-state-minus-pi_2}\\\ref{lem:inverse}}\input{./figures/\fig/\fig_04.tikz}\\
\eq{\ref{lem:inverse}\\\bt}\input{./figures/\fig/\fig_05.tikz}
\eq{\ref{lem:inverse}\\\bo\\\so\\\kt\\\ref{lem:multiplying-global-phases}}\input{./figures/\fig/\fig_06.tikz}
\end{align*}
Then:
\def\fig{parallel-triangles-are-projections-proof-fin}
\begin{align*}
\input{./figures/\fig/\fig_00.tikz}
\eq{\ref{def:triangle}}\input{./figures/\fig/\fig_01.tikz}
\eq{}\input{./figures/\fig/\fig_02.tikz}
\eq{\ref{def:triangle}\\\ref{lem:supp-to-minus-pi_4}}\input{./figures/\fig/\fig_03.tikz}
\eq{\ref{lem:multiplying-global-phases}}\input{./figures/\fig/\fig_04.tikz}
\end{align*}
\qed\end{proof}

\begin{proof}[Lemma \ref{lem:triangles-fork-absorbs-anti-CNOT}]
First:
\def\fig{ug-fork-absorbs-CNOT-proof}
\begin{align*}
\input{./figures/\fig/\fig_00.tikz}
\eq{\ref{def:triangle}\\\so}\input{./figures/\fig/\fig_01.tikz}
\eq{\so\\\ref{lem:inverse}\\\ref{lem:euler-decomp-with-scalar}}\input{./figures/\fig/\fig_02.tikz}
\eq{\ref{lem:C1-bis}\\\so}\input{./figures/\fig/\fig_03.tikz}\\
\eq{\h\\\so}\input{./figures/\fig/\fig_04.tikz}
\eq{\ref{lem:hopf}\\\ref{lem:euler-decomp-with-scalar}}\input{./figures/\fig/\fig_05.tikz}
\eq{\ref{lem:inverse}\\\ref{def:triangle}}\input{./figures/\fig/\fig_06.tikz}
\end{align*}
Then:
\def\fig{ug-fork-absorbs-anti-CNOT-proof}
\begin{align*}
\input{./figures/\fig/\fig_00.tikz}
\eq{\so\\\ref{lem:not-triangle-is-symmetrical}\\\ref{lem:k1}}\input{./figures/\fig/\fig_01.tikz}
\eq{\so\\\ref{lem:not-triangle-is-symmetrical}\\\ref{lem:k1}}\input{./figures/\fig/\fig_02.tikz}
\eq{}\input{./figures/\fig/\fig_03.tikz}
\eq{\so\\\ref{lem:not-triangle-is-symmetrical}\\\ref{lem:k1}}\input{./figures/\fig/\fig_04.tikz}
\end{align*}
\qed\end{proof}

\begin{proof}[Lemma \ref{lem:n-five}]
First:
\def\fig{2-diagrams-of-control-triangle-from-simplified-BW-1}
\begin{align*}
\input{./figures/\fig/\fig_00.tikz}
\eq{\st\\\so\\\ref{lem:not-ug-is-symmetrical}}\input{./figures/\fig/\fig_01.tikz}
\eq{\h}\input{./figures/\fig/\fig_02.tikz}
\eq{\ref{lem:symmetric-diagram-with-triangle-hadamard}}\input{./figures/\fig/\fig_03.tikz}
\eq{\ref{lem:k1}\\\ref{lem:not-triangle-is-symmetrical}}\input{./figures/\fig/\fig_04.tikz}
\eq{\st\\\bo\\\ref{lem:inverse}}\input{./figures/\fig/\fig_05.tikz}
\eq{\ref{lem:n-four}}\input{./figures/\fig/\fig_06.tikz}
\end{align*}
Moreover, from \ref{lem:symmetric-diagram-with-triangle-hadamard}, we can easily derive:
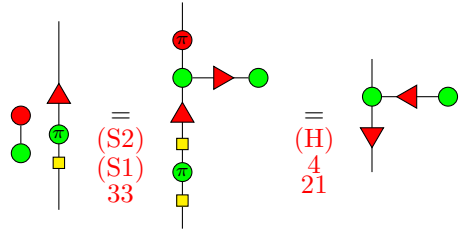
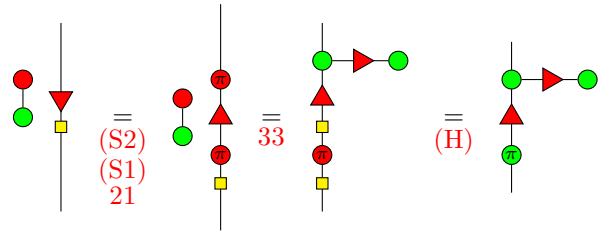
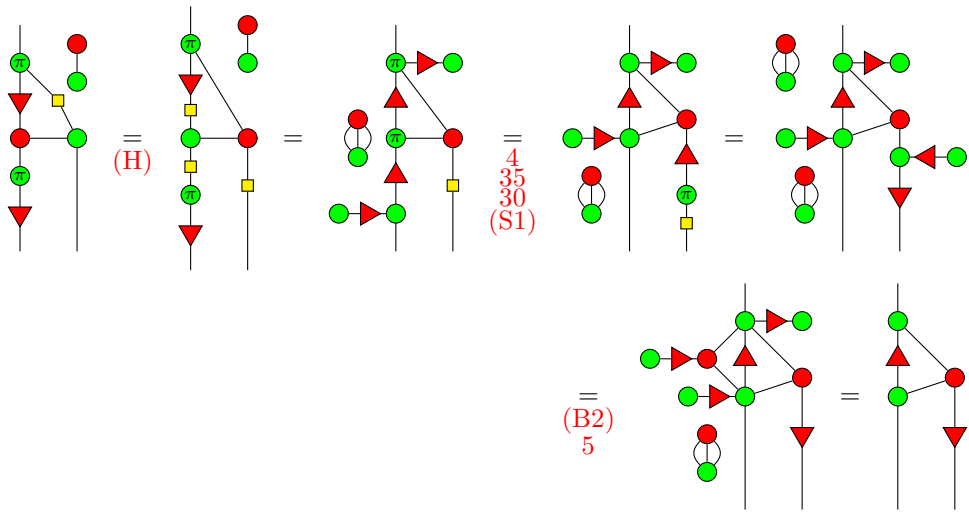
\def\fig{2-diagrams-of-control-triangle-from-simplified-BW-2}
\begin{align*}
\input{./figures/\fig/\fig_00.tikz}
\eq{\st\\\so\\\ref{lem:symmetric-diagram-with-triangle-hadamard}}\input{./figures/\fig/\fig_01.tikz}
\eq{\h\\\ref{lem:k1}\\\ref{lem:not-triangle-is-symmetrical}}\input{./figures/\fig/\fig_02.tikz}
\end{align*}
and:
\def\fig{2-diagrams-of-control-triangle-from-simplified-BW-3}
\begin{align*}
\input{./figures/\fig/\fig_00.tikz}
\eq{\st\\\so\\\ref{lem:not-triangle-is-symmetrical}}\input{./figures/\fig/\fig_01.tikz}
\eq{\ref{lem:symmetric-diagram-with-triangle-hadamard}}\input{./figures/\fig/\fig_02.tikz}
\eq{\h}\input{./figures/\fig/\fig_03.tikz}
\end{align*}
Finally:
\def\fig{2-diagrams-of-control-triangle-from-simplified-BW-4}
\begin{align*}
\input{./figures/\fig/\fig_00.tikz}
\eq{\h}\input{./figures/\fig/\fig_01.tikz}
\eq{}\input{./figures/\fig/\fig_02.tikz}
\eq{\ref{lem:k1}\\\ref{lem:n-four}\\\ref{lem:black-dot-swappable-outputs}\\\so}\input{./figures/\fig/\fig_03.tikz}
\eq{}\input{./figures/\fig/\fig_04.tikz}\\
\eq{\bt\\\ref{lem:inverse}}\input{./figures/\fig/\fig_05.tikz}
\eq{}\input{./figures/\fig/\fig_06.tikz}
\end{align*}
\qed\end{proof}
%
%

\subsection{Proof of Propositions \ref{prop:rules-preserved} and \ref{prop:left-inverse-ZX}}

We first derive an easy but useful lemma for the following:
\begin{lemma}
\label{lem:arity-1-black-dot}
As shown in Section \ref{section:ZWZX}:
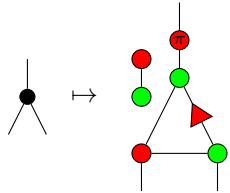
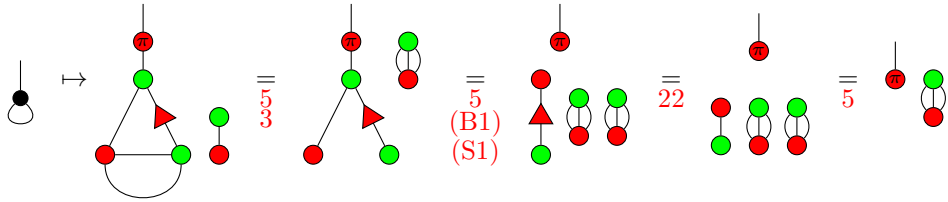
\[
\InputIfFileExists{ZW-to-ZX-dot-1-2-simplified.tikz}{}{\input{./figures/ZW-to-ZX-dot-1-2-simplified.tikz}}
\]
Then:
\def\fig{W-1-0}
\begin{align*}
\input{./figures/\fig/\fig_00.tikz}\mapsto\input{./figures/\fig/\fig_01.tikz}
\eq{\ref{lem:inverse}\\\ref{lem:hopf}}\input{./figures/\fig/\fig_02.tikz}
\eq{\ref{lem:inverse}\\\bo\\\so}\input{./figures/\fig/\fig_03.tikz}
\eq{\ref{lem:red-state-on-triangle}}\input{./figures/\fig/\fig_04.tikz}
\eq{\ref{lem:inverse}}\input{./figures/\fig/\fig_05.tikz}
\end{align*}
\end{lemma}

\subsubsection{Proof of Proposition \ref{prop:rules-preserved}}
\label{prf:rules-preserved}
We prove here that all the rules of the \zwcalc are preserved by $\interpwx{.}$.\\
$\bullet$ X:
\def\fig{rule-X-proof-no-braid}
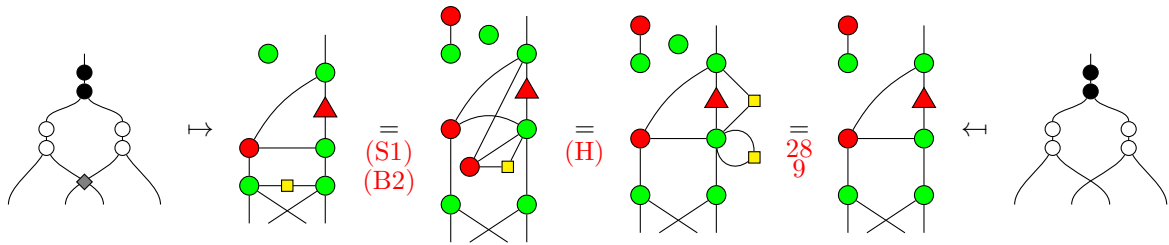
\begin{align*}
\input{./figures/\fig/\fig_00.tikz}~~\mapsto~~\input{./figures/\fig/\fig_01.tikz}
\eq{\so\\\bt}\input{./figures/\fig/\fig_02.tikz}
\eq{\h}\input{./figures/\fig/\fig_03.tikz}
\eq{\ref{lem:triangle-hadamard-parallel}\\\ref{lem:h-loop}}\input{./figures/\fig/\fig_04.tikz}~~\mapsfrom~~\input{./figures/\fig/\fig_05.tikz}
\end{align*}
$\bullet$ $0a$, $0c$, $0d$ and $0d'$ come directly from the paradigm \emph{Only Topology Matters}.\\
$\bullet$ $0b$:
\def\fig{rule-0b-proof-no-braid}
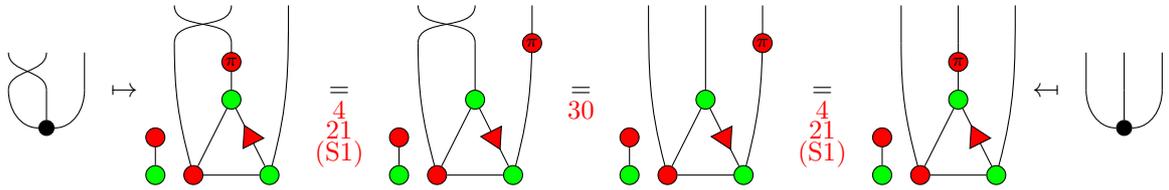
\begin{align*}
\input{./figures/\fig/\fig_00.tikz}~~\mapsto\input{./figures/\fig/\fig_01.tikz}
\eq{\ref{lem:k1}\\\ref{lem:not-triangle-is-symmetrical}\\\so}\input{./figures/\fig/\fig_02.tikz}
\eq{\ref{lem:black-dot-swappable-outputs}}\input{./figures/\fig/\fig_03.tikz}
\eq{\ref{lem:k1}\\\ref{lem:not-triangle-is-symmetrical}\\\so}\input{./figures/\fig/\fig_04.tikz}~\mapsfrom~~\input{./figures/\fig/\fig_05.tikz}
\end{align*}
$\bullet$ $0b'$: Using the result for rule $0b$,
\def\fig{rule-0bp-proof-no-braid}
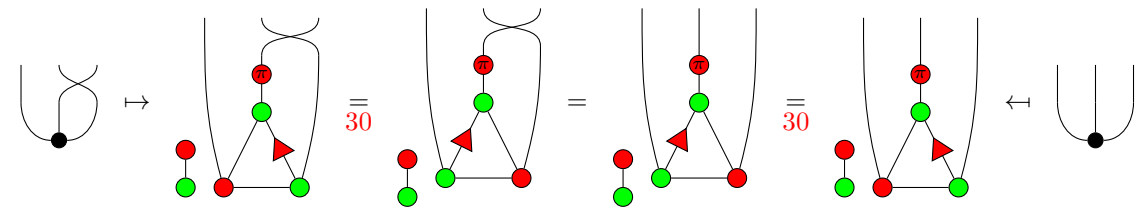
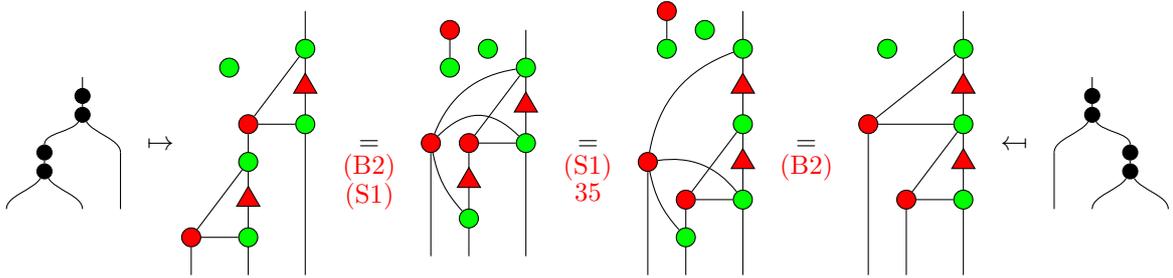
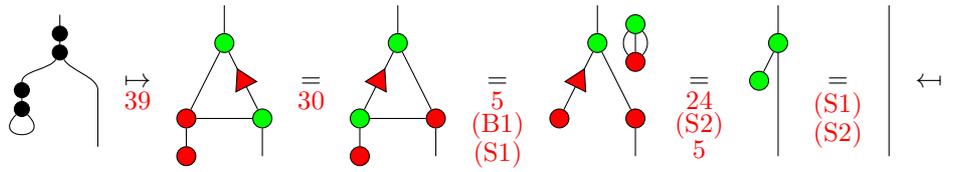
\begin{align*}
\input{./figures/\fig/\fig_00.tikz}~~\mapsto~~\input{./figures/\fig/\fig_01.tikz}
\eq{\ref{lem:black-dot-swappable-outputs}}\input{./figures/\fig/\fig_02.tikz}
\eq{}\input{./figures/\fig/\fig_03.tikz}
\eq{\ref{lem:black-dot-swappable-outputs}}\input{./figures/\fig/\fig_04.tikz}~~\mapsfrom~~\input{./figures/\fig/\fig_05.tikz}
\end{align*}
$\bullet$ $1a$:
\def\fig{rule-1a-proof}
\begin{align*}
\input{./figures/\fig/\fig_00.tikz}~~\mapsto\input{./figures/\fig/\fig_01.tikz}
\eq{\bt\\\so}\input{./figures/\fig/\fig_02.tikz}
\eq{\so\\\ref{lem:n-four}}\input{./figures/\fig/\fig_03.tikz}
\eq{\bt}\input{./figures/\fig/\fig_04.tikz}~~\mapsfrom~~\input{./figures/\fig/\fig_05.tikz}
\end{align*}
$\bullet$ $1b$:
\def\fig{rule-1b-proof}
\begin{align*}
\input{./figures/\fig/\fig_00.tikz}~~\underset{\ref{lem:arity-1-black-dot}}{\mapsto}~~\input{./figures/\fig/\fig_01.tikz}
\eq{\ref{lem:black-dot-swappable-outputs}}\input{./figures/\fig/\fig_02.tikz}
\eq{\ref{lem:inverse}\\\bo\\\so}\input{./figures/\fig/\fig_03.tikz}
\eq{\ref{lem:red-state-on-upside-down-triangle}\\\st\\\ref{lem:inverse}}\input{./figures/\fig/\fig_04.tikz}
\eq{\so\\\st}\input{./figures/\fig/\fig_05.tikz}~~\mapsfrom~~\input{./figures/\fig/\fig_05.tikz}
\end{align*}
$\bullet$ $1c$, $1d$, $2a$ and $2b$ come directly from the spider
rules \so and \st.\\
$\bullet$ $3a$ is the expression of the colour-swapped version of Lemma \ref{lem:k1}.\\
$\bullet$ $3b$:
\def\fig{rule-3b-proof}
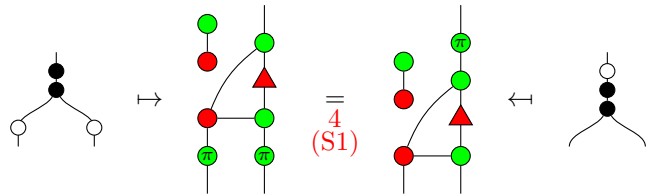
\begin{align*}
\input{./figures/\fig/\fig_00.tikz}\quad\mapsto\quad\input{./figures/\fig/\fig_01.tikz}
\eq[\quad]{\ref{lem:k1}\\\so}\input{./figures/\fig/\fig_02.tikz}\quad\mapsfrom\quad\input{./figures/\fig/\fig_03.tikz}
\end{align*}
$\bullet$ $4$ comes from the spider rule \so.\\
$\bullet$ $5a$: We will need a few steps to prove this equality.\\
\step \label{step:one}
\def\fig{anti-control-pi-and-control-triangle-commute}
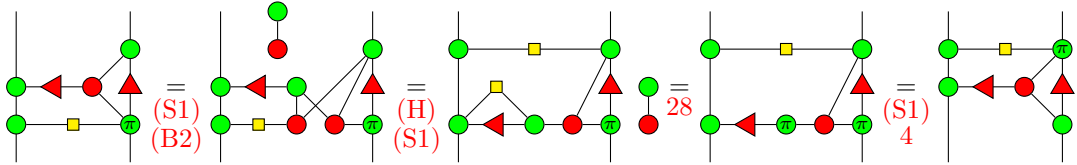
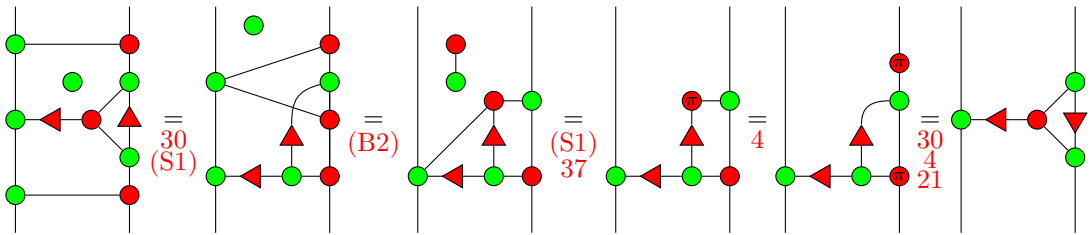
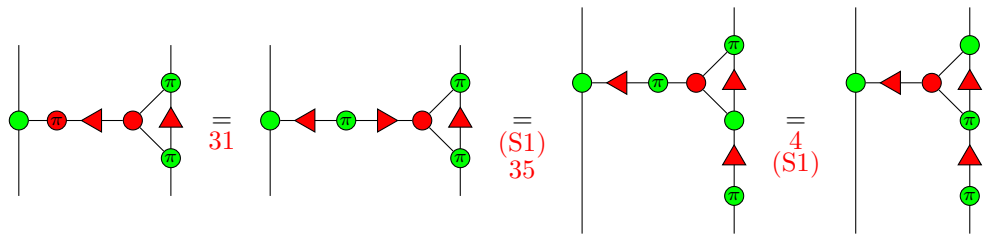
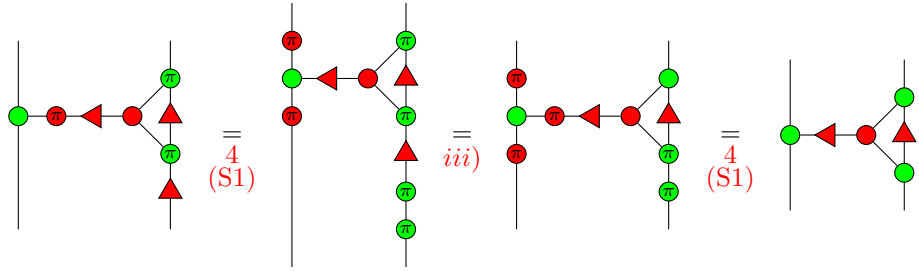
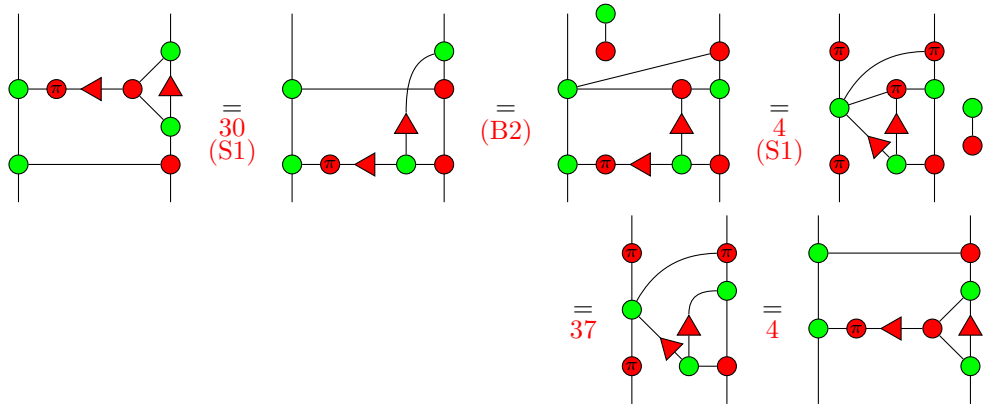
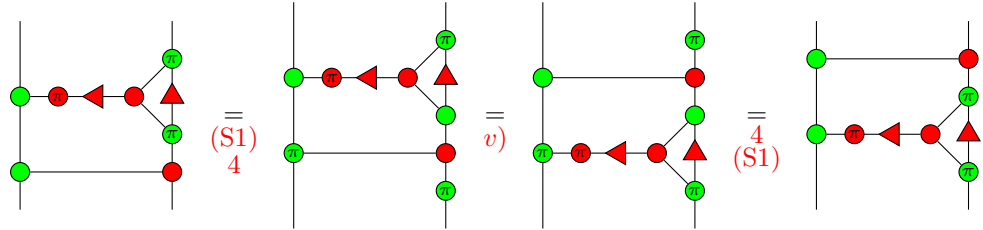
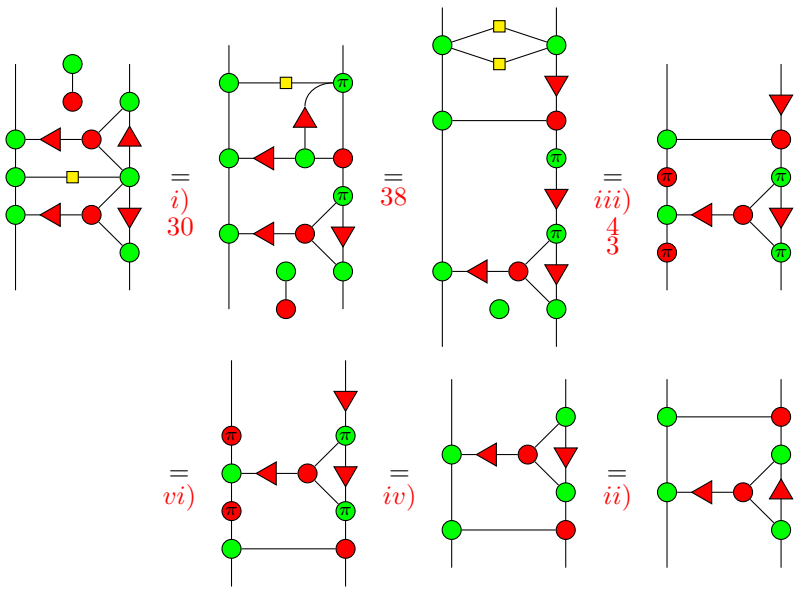
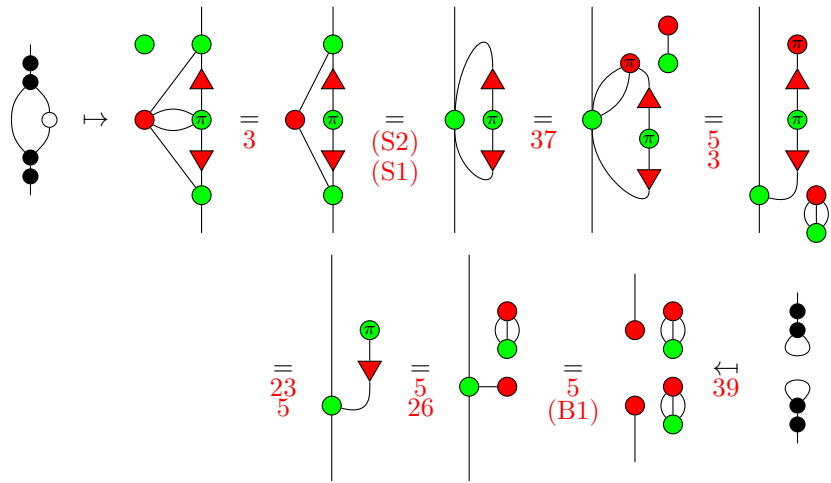
\begin{align*}
\input{./figures/\fig/\fig_00.tikz}
\eq[]{\so\\\bt}\input{./figures/\fig/\fig_01.tikz}
\eq[]{\h\\\so}\input{./figures/\fig/\fig_02.tikz}
\eq[]{\ref{lem:triangle-hadamard-parallel}}\input{./figures/\fig/\fig_03.tikz}
\eq[]{\so\\\ref{lem:k1}}\input{./figures/\fig/\fig_04.tikz}
\end{align*}
\step \label{step:two}
\def\fig{CNOT-control-triangle-CNOT-is-control-transpose-triangle}
\begin{align*}
\input{./figures/\fig/\fig_00.tikz}
\eq[]{\ref{lem:black-dot-swappable-outputs}\\\so}\input{./figures/\fig/\fig_01.tikz}
\eq[]{\bt}\input{./figures/\fig/\fig_02.tikz}
\eq[]{\so\\\ref{lem:triangles-fork-absorbs-anti-CNOT}}\input{./figures/\fig/\fig_03.tikz}
\eq[]{\ref{lem:k1}}\input{./figures/\fig/\fig_04.tikz}
\eq[]{\ref{lem:black-dot-swappable-outputs}\\\ref{lem:k1}\\\ref{lem:not-triangle-is-symmetrical}}\input{./figures/\fig/\fig_05.tikz}
\end{align*}
\step \label{step:three}
\def\fig{anti-control-inverse-triangle-is-control-triangle-times-inverse-triangle}
\begin{align*}
\input{./figures/\fig/\fig_00.tikz}
\eq{\ref{lem:not-ug-is-symmetrical}}\input{./figures/\fig/\fig_01.tikz}
\eq{\so\\\ref{lem:n-four}}\input{./figures/\fig/\fig_02.tikz}
\eq{\ref{lem:k1}\\\so}\input{./figures/\fig/\fig_03.tikz}
\end{align*}
\step \label{step:three-bis}
\def\fig{anti-control-inverse-triangle-is-control-triangle-times-inverse-triangle-bis}
\begin{align*}
\input{./figures/\fig/\fig_00.tikz}
\eq{\ref{lem:k1}\\\so}\input{./figures/\fig/\fig_01.tikz}
\eq{\ref{step:three}}\input{./figures/\fig/\fig_02.tikz}
\eq{\ref{lem:k1}\\\so}\input{./figures/\fig/\fig_03.tikz}
\end{align*}
\step \label{step:four}
\def\fig{anti-control-triangle-and-CNOT-commute}
\begin{align*}
\input{./figures/\fig/\fig_00.tikz}
\eq{\ref{lem:black-dot-swappable-outputs}\\\so}\input{./figures/\fig/\fig_01.tikz}
\eq{\bt}\input{./figures/\fig/\fig_02.tikz}
\eq{\ref{lem:k1}\\\so}\input{./figures/\fig/\fig_03.tikz}\\
\eq{\ref{lem:triangles-fork-absorbs-anti-CNOT}}\input{./figures/\fig/\fig_04.tikz}
\eq{\ref{lem:k1}}\input{./figures/\fig/\fig_05.tikz}
\end{align*}
\step \label{step:five}
\def\fig{anti-control-inverse-triangle-and-CNOT-commute}
\begin{align*}
\input{./figures/\fig/\fig_00.tikz}
\eq{\so\\\ref{lem:k1}}\input{./figures/\fig/\fig_01.tikz}
\eq{\ref{step:four}}\input{./figures/\fig/\fig_02.tikz}
\eq{\ref{lem:k1}\\\so}\input{./figures/\fig/\fig_03.tikz}
\end{align*}
\step \label{step:six}
\def\fig{control-not-ug-is-symmetrical-proof}
\begin{align*}
\input{./figures/\fig/\fig_00.tikz}
\eq{\ref{step:one}\\\ref{lem:black-dot-swappable-outputs}}\input{./figures/\fig/\fig_01.tikz}
\eq{\ref{lem:n-five}}\input{./figures/\fig/\fig_02.tikz}
\eq{\ref{step:three}\\\ref{lem:k1}\\\ref{lem:hopf}}\input{./figures/\fig/\fig_03.tikz}\\
\eq{\ref{step:five}}\input{./figures/\fig/\fig_04.tikz}
\eq{\ref{step:three-bis}}\input{./figures/\fig/\fig_05.tikz}
\eq{\ref{step:two}}\input{./figures/\fig/\fig_06.tikz}
\end{align*}
\step \label{step:seven}
\def\fig{2-diagrams-of-control-triangle}
\begin{align*}
\input{./figures/\fig/\fig_00.tikz}
\eq{\so\\\ref{lem:k1}\\\ref{lem:not-triangle-is-symmetrical}}\input{./figures/\fig/\fig_01.tikz}
\eq{\bt}\input{./figures/\fig/\fig_02.tikz}
\eq{\so}\input{./figures/\fig/\fig_03.tikz}
\eq{\bt}\input{./figures/\fig/\fig_04.tikz}\\
\eq{\ref{lem:looped-triangle}\\\ref{lem:black-dot-swappable-outputs}}\input{./figures/\fig/\fig_05.tikz}
\eq{\ref{lem:triangles-fork-absorbs-anti-CNOT}}\input{./figures/\fig/\fig_06.tikz}
\eq{\bt}\input{./figures/\fig/\fig_07.tikz}
\eq{\ref{lem:looped-triangle}\\\ref{lem:black-dot-swappable-outputs}}\input{./figures/\fig/\fig_08.tikz}
\end{align*}
Finally,
\def\fig{rule-5a-proof}
\begin{align*}
\input{./figures/\fig/\fig_00.tikz}~~\mapsto~~\input{./figures/\fig/\fig_01.tikz}
\eq{}\input{./figures/\fig/\fig_02.tikz}
\eq{\ref{lem:control-pi-and-anti-CNOT-commute}}\input{./figures/\fig/\fig_03.tikz}
\eq{\ref{step:seven} \\\ref{lem:hopf}}\input{./figures/\fig/\fig_04.tikz}\\
\eq{\ref{step:six}}\input{./figures/\fig/\fig_05.tikz}
\eq{\so\\\bt}\input{./figures/\fig/\fig_06.tikz}
\eq{\so\\\ref{step:seven}}\input{./figures/\fig/\fig_07.tikz}~~\mapsfrom\!\!\input{./figures/\fig/\fig_08.tikz}=\input{./figures/\fig/\fig_09.tikz}
\end{align*}
$\bullet$ $5b$:
\def\fig{rule-5b-proof}
\begin{align*}
\input{./figures/\fig/\fig_00.tikz}~~\underset{\ref{lem:arity-1-black-dot}}{\mapsto}~~\input{./figures/\fig/\fig_01.tikz}
\eq{\ref{lem:inverse}\\\bo\\\so}\input{./figures/\fig/\fig_02.tikz}
\eq{\ref{lem:red-state-on-triangle}}\input{./figures/\fig/\fig_03.tikz}
\eq{\ref{lem:inverse}\\\bo}\input{./figures/\fig/\fig_04.tikz}~~\underset{\ref{lem:arity-1-black-dot}}{\mapsfrom}~~\input{./figures/\fig/\fig_05.tikz}
\end{align*}
$\bullet$ $5c$:
\def\fig{rule-5c-proof}
\begin{align*}
\input{./figures/\fig/\fig_00.tikz}~~\underset{\ref{lem:arity-1-black-dot}}{\mapsto}~~\input{./figures/\fig/\fig_01.tikz}
\eq{\so\\\ref{lem:2-is-sqrt2-squared}}\input{./figures/\fig/\fig_02.tikz}
\eq{\ref{lem:inverse}}\input{./figures/\fig/\fig_03.tikz}~~\mapsfrom~~\input{./figures/\fig/\fig_03.tikz}
\end{align*}
$\bullet$ $5d$:
\def\fig{rule-5d-proof}
\begin{align*}
\input{./figures/\fig/\fig_00.tikz}~~\mapsto~~\input{./figures/\fig/\fig_01.tikz}
\eq{\ref{lem:hopf}}\input{./figures/\fig/\fig_02.tikz}
\eq{\st\\\so}\input{./figures/\fig/\fig_03.tikz}
\eq{\ref{lem:triangles-fork-absorbs-anti-CNOT}}\input{./figures/\fig/\fig_04.tikz}
\eq{\ref{lem:inverse}\\\ref{lem:hopf}}\input{./figures/\fig/\fig_05.tikz}\\
\eq{\ref{lem:pi-red-state-on-triangle}\\\ref{lem:inverse}}\input{./figures/\fig/\fig_06.tikz}
\eq{\ref{lem:inverse}\\\ref{lem:pi-green-state-on-upside-down-triangle}}\input{./figures/\fig/\fig_07.tikz}
\eq{\ref{lem:inverse}\\\bo}\input{./figures/\fig/\fig_08.tikz}~~\underset{\ref{lem:arity-1-black-dot}}{\mapsfrom}~~\input{./figures/\fig/\fig_09.tikz}
\end{align*}
$\bullet$ $6a$: Thanks to the rule X we can get rid of \resizebox{!}{1em}{
\InputIfFileExists{control-pi.tikz}{}{\input{./figures/control-pi.tikz}}
} induced by the crossing. Then,
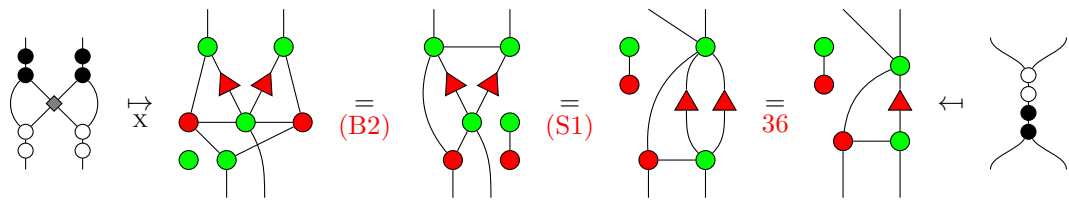
\def\fig{rule-6a-proof-no-braid}
\begin{align*}
\input{./figures/\fig/\fig_00.tikz}~~\underset{\text{X}}{\mapsto}~~\input{./figures/\fig/\fig_01.tikz}
\eq{\bt}\input{./figures/\fig/\fig_02.tikz}
\eq{\so}\input{./figures/\fig/\fig_03.tikz}
\eq{\ref{lem:parallel-triangles}}\input{./figures/\fig/\fig_04.tikz}~~\mapsfrom~~\input{./figures/\fig/\fig_05.tikz}
\end{align*}
$\bullet$ $6b$ is exactly the copy rule \bo.\\
$\bullet$ $6c$:
\def\fig{rule-6c-proof}
\begin{align*}
\input{./figures/\fig/\fig_00.tikz}~~\mapsto~~\input{./figures/\fig/\fig_01.tikz}
\eq{\so\\\ref{lem:inverse}\\\ref{lem:hopf}}\input{./figures/\fig/\fig_02.tikz}
\eq{\ref{lem:inverse}\\\bo\\\st}\input{./figures/\fig/\fig_03.tikz}
\eq{\ref{lem:red-state-on-triangle}}\input{./figures/\fig/\fig_04.tikz}~~\underset{\ref{lem:arity-1-black-dot}}{\mapsfrom}\input{./figures/\fig/\fig_05.tikz}
\end{align*}
$\bullet$ $7a$:
\def\fig{rule-7a-proof-no-braid}
\begin{align*}
\input{./figures/\fig/\fig_00.tikz}~~\mapsto~~\input{./figures/\fig/\fig_01.tikz}
\eq[]{\so\\\h}\input{./figures/\fig/\fig_02.tikz}
\eq[]{\bt}\input{./figures/\fig/\fig_03.tikz}
\eq[]{\h}\input{./figures/\fig/\fig_04.tikz}~~\mapsfrom~~\input{./figures/\fig/\fig_05.tikz}
\end{align*}
$\bullet$ $7b$: using \ref{lem:k1}, \h and \so:
\def\fig{rule-7b-proof-no-braid}
\begin{align*}
\input{./figures/\fig/\fig_00.tikz}~~\mapsto~~\input{./figures/\fig/\fig_01.tikz}
\eq{\ref{lem:k1}\\\h\\\so}\input{./figures/\fig/\fig_02.tikz}
~~\mapsfrom~~\input{./figures/\fig/\fig_04.tikz}
\end{align*}
$\bullet$ R$_2$:
\def\fig{rule-r2-proof}
\begin{align*}
\input{./figures/\fig/\fig_00.tikz}~~\mapsto~~\input{./figures/\fig/\fig_01.tikz}
\eq{\so}\input{./figures/\fig/\fig_02.tikz}
\eq{\ref{lem:hopf}}\input{./figures/\fig/\fig_03.tikz}~~\mapsfrom~~\input{./figures/\fig/\fig_03.tikz}
\end{align*}
$\bullet$ R$_3$:
\def\fig{rule-r3-proof}
\begin{align*}
\input{./figures/\fig/\fig_00.tikz}~~\mapsto~~\input{./figures/\fig/\fig_01.tikz}
\eq{\so}\input{./figures/\fig/\fig_02.tikz}
\eq{\so}\input{./figures/\fig/\fig_03.tikz}~~\mapsfrom~~\input{./figures/\fig/\fig_04.tikz}
\end{align*}
$\bullet$ $iv$: using \sth, \so, \ref{lem:2-is-sqrt2-squared} and \ref{lem:inverse},
\def\fig{rule-half-proof}
\begin{align*}
\input{./figures/\fig/\fig_00.tikz}~~\mapsto~~\input{./figures/\fig/\fig_01.tikz}
\eq{\sth}\input{./figures/\fig/\fig_02.tikz}
\eq{\so}\input{./figures/\fig/\fig_03.tikz}
\eq{\ref{lem:2-is-sqrt2-squared}\\\ref{lem:inverse}}\input{./figures/\fig/\fig_04.tikz}
~~\mapsfrom~~\input{./figures/\fig/\fig_05.tikz}
\end{align*}
\qed

\subsubsection{Proof of Proposition \ref{prop:left-inverse-ZX}}
\label{prf:left-inverse-ZX}

Let us write $\dblinterp{.}=\interpwx{\interpxw{.}}$.
We can show inductively that:
\[\zxt\vdash \dblinterp{D}\circ\left(
\InputIfFileExists{top-composition.tikz}{}{\input{./figures/top-composition.tikz}}
\right) = D\otimes \left(
\InputIfFileExists{theta.tikz}{}{\input{./figures/theta.tikz}}
\right)\]
which is the expression of Proposition \ref{prop:left-inverse}.\\
$\bullet$ The result is obvious for the generators 
\InputIfFileExists{empty-diagram.tikz}{}{\input{./figures/empty-diagram.tikz}}
, 
\begin{tikzpicture}
	\begin{pgfonlayer}{nodelayer}
		\node [style=none] (0) at (0, 0.2499999) {};
		\node [style=none] (1) at (0, -0.2499999) {};
	\end{pgfonlayer}
	\begin{pgfonlayer}{edgelayer}
		\draw (0.center) to (1.center);
	\end{pgfonlayer}
\end{tikzpicture}}
, 
\InputIfFileExists{crossing.tikz}{}{\input{./figures/crossing.tikz}}
, 
\begin{tikzpicture}
	\begin{pgfonlayer}{nodelayer}
		\node [style=none] (0) at (-0.2500001, -0) {};
		\node [style=none] (1) at (0.2500001, -0) {};
	\end{pgfonlayer}
	\begin{pgfonlayer}{edgelayer}
		\draw [bend left=90, looseness=1.75] (0.center) to (1.center);
	\end{pgfonlayer}
\end{tikzpicture}}
, and 
\begin{tikzpicture}
	\begin{pgfonlayer}{nodelayer}
		\node [style=none] (0) at (-0.2500001, 0.2500001) {};
		\node [style=none] (1) at (0.2500001, 0.2500001) {};
	\end{pgfonlayer}
	\begin{pgfonlayer}{edgelayer}
		\draw [bend right=90, looseness=1.75] (0.center) to (1.center);
	\end{pgfonlayer}
\end{tikzpicture}}
.\\
$\bullet$ 
\begin{tikzpicture}
	\begin{pgfonlayer}{nodelayer}
		\node [style={H box}] (0) at (0, 0) {};
		\node [style=none] (1) at (0, 0.5) {};
		\node [style=none] (2) at (0, -0.5) {};
	\end{pgfonlayer}
	\begin{pgfonlayer}{edgelayer}
		\draw (2.center) to (1.center);
	\end{pgfonlayer}
\end{tikzpicture}
}
:
\def\fig{hadamard-double-interpretation}
\begin{align*}
\interpwx{~
\InputIfFileExists{hadamard-interpretation-2-no-braid.tikz}{}{\input{./figures/hadamard-interpretation-2-no-braid.tikz}}
~}\hspace{-1em} \eq{} \input{./figures/\fig/\fig_00.tikz}
\eq{\so\\\ref{lem:inverse}}\input{./figures/\fig/\fig_01.tikz}
\eq{\h\\\ref{lem:k1}}\input{./figures/\fig/\fig_02.tikz}
\end{align*}
and, using \so, \eu, \ref{lem:inverse}, \h, \ref{lem:hopf} and \ref{lem:green-state-pi_2-is-red-state-minus-pi_2}:
\def\fig{sqrt2-double-interp-times-theta}
\begin{align*}
\input{./figures/\fig/\fig_00.tikz}
\eq{\kt\\\st\\\so\\\eu}\input{./figures/\fig/\fig_01.tikz}
\eq{\h\\\so}\input{./figures/\fig/\fig_02.tikz}
\eq{\ref{lem:inverse}\\\ref{lem:hopf}}\input{./figures/\fig/\fig_03.tikz}
\eq{\ref{lem:green-state-pi_2-is-red-state-minus-pi_2}}\input{./figures/\fig/\fig_04.tikz}
\end{align*}
Hence $ ZX\vdash \dblinterp{
\begin{tikzpicture}
	\begin{pgfonlayer}{nodelayer}
		\node [style={H box}] (0) at (0, 0) {};
		\node [style=none] (1) at (0, 0.5) {};
		\node [style=none] (2) at (0, -0.5) {};
	\end{pgfonlayer}
	\begin{pgfonlayer}{edgelayer}
		\draw (2.center) to (1.center);
	\end{pgfonlayer}
\end{tikzpicture}
}
}\hspace{0em}\circ\left(
\InputIfFileExists{top-composition-1.tikz}{}{\input{./figures/top-composition-1.tikz}}
\right) ~=~~ 
\begin{tikzpicture}
	\begin{pgfonlayer}{nodelayer}
		\node [style={H box}] (0) at (0, 0) {};
		\node [style=none] (1) at (0, 0.5) {};
		\node [style=none] (2) at (0, -0.5) {};
	\end{pgfonlayer}
	\begin{pgfonlayer}{edgelayer}
		\draw (2.center) to (1.center);
	\end{pgfonlayer}
\end{tikzpicture}
}
\hspace{0.1em}\raisebox{-0.7em}{
\InputIfFileExists{theta.tikz}{}{\input{./figures/theta.tikz}}
}$\\~~\\\\
$\bullet \interpwx{~
\InputIfFileExists{gn-0-1-m-interpretation.tikz}{}{\input{./figures/gn-0-1-m-interpretation.tikz}}
~~}\hspace{-1.5em}\circ\left(
\InputIfFileExists{top-composition-1.tikz}{}{\input{./figures/top-composition-1.tikz}}
\right) = 
\InputIfFileExists{gn-0-1-m-theta.tikz}{}{\input{./figures/gn-0-1-m-theta.tikz}}
$\\~~\\\\
$\bullet \interpwx{~
\InputIfFileExists{gn-0-n-1-interpretation.tikz}{}{\input{./figures/gn-0-n-1-interpretation.tikz}}
~~}\hspace{-1.5em}\circ\left(
\InputIfFileExists{top-composition.tikz}{}{\input{./figures/top-composition.tikz}}
\right) = 
\InputIfFileExists{gn-0-n-1-theta.tikz}{}{\input{./figures/gn-0-n-1-theta.tikz}}
$\\~~\\\\
$\bullet$ 
\begin{tikzpicture}
	\begin{pgfonlayer}{nodelayer}
		\node [style=gn] (0) at (0, -0) {$\frac{\pi}{4}$};
		\node [style=none] (1) at (0, 0.5000001) {};
		\node [style=none] (2) at (0, -0.5000001) {};
	\end{pgfonlayer}
	\begin{pgfonlayer}{edgelayer}
		\draw (1.center) to (2.center);
	\end{pgfonlayer}
\end{tikzpicture}}
:
\def\fig{gn-pi_4-double-interpretation}
\begin{align*}
\interpwx{
\InputIfFileExists{gn-pi_4-interpretation-2-no-braid.tikz}{}{\input{./figures/gn-pi_4-interpretation-2-no-braid.tikz}}
} \!\!\!\!\eq{} \input{./figures/\fig/\fig_00.tikz}
\eq{\ref{lem:inverse}\\\so\\\h\\\ref{lem:k1}\\\bo}\input{./figures/\fig/\fig_01.tikz}
\end{align*}
But:
\def\fig{gn-state-pi_2-on-CRG-is-control-pi_2}
\begin{align*}
\input{./figures/\fig/\fig_00.tikz}
\eq{\ref{def:triangle}}\input{./figures/\fig/\fig_01.tikz}
\eq{\ref{lem:supp-to-minus-pi_4}\\\ref{lem:euler-decomp-with-scalar}}\input{./figures/\fig/\fig_02.tikz}
\eq{\h\\\so}\input{./figures/\fig/\fig_03.tikz}\\
\eq{\ref{lem:C1-original}}\input{./figures/\fig/\fig_04.tikz}
\eq{\h\\\bo\\\so\\\st\\\kt}\input{./figures/\fig/\fig_05.tikz}
\end{align*}
So that:
\def\fig{gn-pi_4-double-interp-times-theta}
\begin{align*}
\input{./figures/\fig/\fig_00.tikz}
\eq{\so\\\eu}\input{./figures/\fig/\fig_01.tikz}
\eq{\ref{lem:hopf}\\\ref{lem:not-triangle-is-symmetrical}}\input{./figures/\fig/\fig_02.tikz}
\eq{}\input{./figures/\fig/\fig_03.tikz}\\
\eq{\so\\\bt}\input{./figures/\fig/\fig_04.tikz}
\eq{\so}\input{./figures/\fig/\fig_05.tikz}
\eq{\bo\\\so}\input{./figures/\fig/\fig_06.tikz}
\end{align*}
which means $ \zxt\vdash \dblinterp{
\begin{tikzpicture}
	\begin{pgfonlayer}{nodelayer}
		\node [style=gn] (0) at (0, -0) {$\frac{\pi}{4}$};
		\node [style=none] (1) at (0, 0.5000001) {};
		\node [style=none] (2) at (0, -0.5000001) {};
	\end{pgfonlayer}
	\begin{pgfonlayer}{edgelayer}
		\draw (1.center) to (2.center);
	\end{pgfonlayer}
\end{tikzpicture}}
}\circ\left(
\InputIfFileExists{top-composition-1.tikz}{}{\input{./figures/top-composition-1.tikz}}
\right) ~=~~ 
\begin{tikzpicture}
	\begin{pgfonlayer}{nodelayer}
		\node [style=gn] (0) at (0, -0) {$\frac{\pi}{4}$};
		\node [style=none] (1) at (0, 0.5000001) {};
		\node [style=none] (2) at (0, -0.5000001) {};
	\end{pgfonlayer}
	\begin{pgfonlayer}{edgelayer}
		\draw (1.center) to (2.center);
	\end{pgfonlayer}
\end{tikzpicture}}
\hspace{0.1em}\raisebox{-0.7em}{
\InputIfFileExists{theta.tikz}{}{\input{./figures/theta.tikz}}
}$\\~~\\\\
$\bullet$ $D_1\circ D_2$:\\
It is to be noticed that $\interpwx{D_1\circ D_2}=\interpwx{D_1}\circ \interpwx{D_2}$ and $\interpwx{D_1\otimes D_2}=\interpwx{D_1}\otimes \interpwx{D_2}$.\\ Let us write $\theta = 
\InputIfFileExists{theta.tikz}{}{\input{./figures/theta.tikz}}
$. Then:
\begin{align*}
\zxt\vdash \dblinterp{D_1\circ D_2}\circ(\mathbb{I}\otimes \theta) &= \dblinterp{D_1}\circ \dblinterp{D_2}\circ(\mathbb{I}\otimes \theta)=  \dblinterp{D_1}\circ (D_2\otimes \theta) \\&= \dblinterp{D_1}\circ(\mathbb{I}\otimes \theta)\circ D_2 = (D_1\otimes\theta)\circ D_2 \\
&= (D_1\circ D_2)\otimes \theta
\end{align*}
$\bullet$ $D_1\otimes D_2$:
\begin{align*}
\zxt\vdash&\dblinterp{D_1\otimes D_2}\circ(\mathbb{I}\otimes \theta)\\
&=\left( \interpwx{\mathbb{I}^{\otimes n'}}\hspace{-0.8em}\otimes\dblinterp{D_2}\right)\circ\interpwx{\nmcrossI{m}{n'}}\hspace{-1.6em} \circ \left(\interpwx{\mathbb{I}^{\otimes m}\vphantom{\rule{1pt}{1.2em}}}\hspace{-0.8em}\otimes\dblinterp{D_1}\right)\circ\interpwx{\nmcrossI{n}{m}}\hspace{-1.6em}\circ(\mathbb{I}\otimes \theta)\\
&= \left( \mathbb{I}^{\otimes n'}\otimes\dblinterp{D_2}\right)\circ\left(\nmcrossI{m}{n'}\right)\circ \left(\mathbb{I}^{\otimes m}\otimes\dblinterp{D_1}\right)\circ\left(\nmcrossI{n}{m}\right)\circ(\mathbb{I}\otimes \theta)\\
&=\left( \mathbb{I}^{\otimes n'}\otimes\dblinterp{D_2}\right)\circ\left(\nmcrossI{m}{n'}\right)\circ \left(\mathbb{I}^{\otimes m}\otimes(\dblinterp{D_1}\circ(\mathbb{I}\otimes \theta))\right)\circ\left(\nmcross{n}{m}\right)\\
&=\left( \mathbb{I}^{\otimes n'}\otimes\dblinterp{D_2}\right)\circ\left(\nmcrossI{m}{n'}\right)\circ \left(\mathbb{I}^{\otimes m}\otimes D_1\otimes \theta\right)\circ\left(\nmcross{n}{m}\right)\\
&=\left( \mathbb{I}^{\otimes n'}\otimes\dblinterp{D_2}\right)\circ\left(\nmcrossI{m}{n'}\right)\circ(\mathbb{I}\otimes \theta)\circ \left(\mathbb{I}^{\otimes m}\otimes D_1\right)\circ\left(\nmcross{n}{m}\right)\\
&=\left( \mathbb{I}^{\otimes n'}\otimes D_2\otimes\theta\right)\circ\left(\nmcross{m}{n'}\right)\circ \left(\mathbb{I}^{\otimes m}\otimes D_1\right)\circ\left(\nmcross{n}{m}\right)\\
&=\left[\left( \mathbb{I}^{\otimes n'}\otimes D_2\right)\circ\left(\nmcross{m}{n'}\right)\circ \left(\mathbb{I}^{\otimes m}\otimes D_1\right)\circ\left(\nmcross{n}{m}\right)\right]\otimes\theta\\
&=D_1\otimes D_2\otimes\theta\\
\end{align*}
By compositions, for any diagram $D$, $\zxt\vdash \dblinterp{D}\circ(\mathbb{I}\otimes \theta) = D\otimes\theta$. Then, using Lemmas \ref{lem:bicolor-0-alpha} and \ref{lem:inverse}:
\begin{align*}
\forall D\in \zxt,\quad
\zxt\vdash \left(\hspace{-0.2em}\scalebox{0.8}{
\InputIfFileExists{bottom-composition.tikz}{}{\input{./figures/bottom-composition.tikz}}
}\hspace{-0.3em}\right)\circ\interpwx{\interpxw{D}}\circ\left(\hspace{-0.2em}\scalebox{0.8}{
\InputIfFileExists{top-composition.tikz}{}{\input{./figures/top-composition.tikz}}
}\hspace{-0.3em}\right)=D
\end{align*}
\qed

\end{document}